\documentclass[submission,copyright,creativecommons]{eptcs}
\providecommand{\event}{AiML 2026} 

\usepackage{aiml26}

\usepackage{amsmath}
\usepackage{amsthm}
\usepackage{amssymb}
\usepackage{colonequals}
\usepackage{proof}
\usepackage{multirow}

\makeatletter
\newtheorem*{rep@theorem}{\rep@title}
\newcommand{\newreptheorem}[2]{%
\newenvironment{rep#1}[1]{%
 \def\rep@title{#2 \ref{##1}}%
 \begin{rep@theorem}}%
 {\end{rep@theorem}}}
\makeatother

\newreptheorem{theorem}{Theorem}

\usepackage{hyperref}

\usepackage{virginialake}
\vlnosmallleftlabels
\vlnostructuressyntax

\usepackage{tikz}
\usepackage[inline]{enumitem}
\usepackage{multicol}

\usepackage{ifthen}
\usepackage{xparse}

\newcommand{\paaras}[1]
{\begin{color}{blue}(Paaras: #1)\end{color}}
\newcommand{\sonia}[1]
{\begin{color}{teal}(Sonia: #1)\end{color}}
\newcommand{\todo}[1]
{\begin{color}{purple}(Todo: #1)\end{color}}

\newcommand{\revpaaras}[1]
{\begin{color}{blue}(Paaras: #1)\end{color}}
\newcommand{\review}[1]
{\begin{color}{red}(Todo: #1)\end{color}}

\newcommand{\tuple}[1]{(#1)}

\newcommand{\AND}{\mathbin{\wedge}}
\newcommand{\OR}{\mathbin{\vee}}
\newcommand{\IMPLIES}{\mathbin{\to}}
\newcommand{\IMP}{\IMPLIES}
\newcommand{\NOT}{\mathord{\neg}}
\newcommand{\DIA}[1][]{\Diamond_{#1}}
\newcommand{\BOX}[1][]{\Box_{#1}}
\newcommand{\bottom}{\mathord{\bot}}
\newcommand{\BOT}{\mathord{\bottom}}

\newcommand{\BNFequals}{\coloncolonequals}
\newcommand{\proves}[3][]{#1 \vdash_{#2} #3}
\newcommand{\notproves}[3][]{#1 \not\vdash_{#2} #3}
\newcommand{\langjust}{\mathcal{L}_{\logic{J}}}
\newcommand{\langmod}{\mathcal{L}_{\BOX}}
\newcommand{\langmodann}{\mathcal{L}_{\BOX}^\mathsf{ann}}

\newcommand{\logic}[1]{\mathsf{#1}}
\newcommand{\PA}{\logic{PA}}
\newcommand{\HA}{\logic{HA}}

\newcommand{\LP}{\logic{LP}}
\newcommand{\iLP}{\logic{iLP}}
\newcommand{\CS}{\logic{CS}}

\newcommand{\JLCS}{\JIKCS}

\newcommand{\JIK}{\logic{JIK}}
\newcommand{\JIKCS}{\JIK_{\CS}}

\newcommand{\sfour}{\logic{S4}}
\newcommand{\isfour}{\logic{iS4}}

\newcommand{\CK}{\logic{CK}}
\newcommand{\IK}{\logic{IK}}
\newcommand{\ISfour}{\logic{IS4}}
\newcommand{\ISfive}{\logic{IS5}}

\newcommand{\IPL}{\logic{IPL}}

\newcommand{\kax}[1][]{\mathsf{{k}_{#1}}}

\newcommand{\taxb}{\mathsf{t}_{\BOX}}
\newcommand{\taxd}{\mathsf{t}_{\DIA}}

\newcommand{\faxb}{\mathsf{4}_{\BOX}}
\newcommand{\faxd}{\mathsf{4}_{\DIA}}

\newcommand{\jkax}[1][]{\mathsf{{jk}_{#1}}}

\newcommand{\jtaxb}{\mathsf{jt}_{\BOX}}
\newcommand{\jtaxd}{\mathsf{jt}_{\DIA}}
\newcommand{\jfaxb}{\mathsf{j4}_{\BOX}}
\newcommand{\jfaxd}{\mathsf{j4}_{\DIA}}
\newcommand{\jsumaxl}{\mathsf{j}{\jsum{}{}}_l}
\newcommand{\jsumaxr}{\mathsf{j}{\jsum{}{}}_r}
\newcommand{\juniaxl}{\mathsf{j}{\suni{}{}}_l}
\newcommand{\juniaxr}{\mathsf{j}{\suni{}{}}_r}

\newcommand{\modus}{\hyperlink{rule:MP}{\mathsf{mp}}}
\newcommand{\axrule}{\hyperlink{rule:ax}{\mathsf{ax}}}
\newcommand{\idrule}{\hyperlink{rule:id}{\mathsf{id}}}
\newcommand{\nec}{\hyperlink{rule:nec}{\mathsf{nec}}}
\newcommand{\can}{\hyperlink{rule:can}{\mathsf{can}}}

\newcommand{\Prop}{\mathsf{Prop}}

\newcommand{\IPLax}{\logic{Int}}
\newcommand{\BaseAx}{\logic{Base}}

\newcommand{\IKax}{\logic{Mod}}
\newcommand{\axinst}[1]{#1^{I}}

\newcommand{\tsign}{\mathsf{T}}
\newcommand{\fsign}{\mathsf{F}}
\newcommand{\T}[1]{\tsign{#1}}
\newcommand{\F}[1]{\fsign{#1}}

\newcommand{\term}[1]{\mathtt{#1}}
\newcommand{\tm}[1]{\term{#1}}
\newcommand{\just}[2]{[\tm{#1}] #2} 
\newcommand{\sat}[2]{\langle \tm{#1} \rangle #2} 
\newcommand{\pt}[1]{\tm {#1}}
\newcommand{\prfterms}{\mathsf{PrfTm}}
\newcommand{\prftm}{\prfterms}
\newcommand{\satterms}{\mathsf{SatTm}}
\newcommand{\prfvars}{\mathsf{PrfVar}}
\newcommand{\satvars}{\mathsf{SatVar}}
\newcommand{\terms}{\mathsf{Tm}}
\newcommand{\pvar}[1][]{\tm {x_{\mathnormal{#1}}}}
\newcommand{\svar}[1][]{\tm {a_{\mathnormal{#1}}}}

\newcommand{\pconst}[1][]{\tm c_{#1}}
\newcommand{\prfconst}{\mathsf{PrfConst}}
\newcommand{\jsum}[2]{\term{#1} \mathbin{+} \term{#2}}
\newcommand{\jappl}[2]{\term{#1} \mathbin{\cdot} \term{#2}}
\newcommand{\bangsymb}{{!}}
\newcommand{\jbang}[2][]{\bangsymb^{\mathnormal{#1}}\term{#2}}
\newcommand{\suni}[2]{\term{#1} \mathbin{\sqcup} \term{#2}}
\newcommand{\sappl}[2]{\term{#1} \mathbin{\star} \term{#2}}
\newcommand{\jupdt}[2]{\term{#1} \mathbin{\triangleright} \term{#2}} 

\newcommand{\set}[2][]{%
	\ifthenelse{\equal{#1}{}}%
	{\{#2\}}%
	{\{#1 \mid #2\}}%
}
\newcommand{\sub}{\sigma}
\newcommand{\subst}[2][\sub]{#2 #1}
\newcommand{\function}[3]{#1 : #2 \rightarrow #3}

\newcommand{\forget}[2][\forgetmap]{{#2}^{#1}}

\newcommand{\nat}{\mathbb{N}}

\newcommand{\UNI}{\cup}

\newcommand{\dom}[1]{\text{dom}(#1)}
\newcommand{\sset}{\subseteq}
\newcommand{\ssetop}{\supseteq}

\newcommand{\pset}[1]{\mathcal{P}(#1)}
\newcommand{\cartprod}{\times}
\newcommand{\inv}[2]{\tm #1^{-1}{#2}}

\newcommand{\modeltext}[1]{\mathcal{#1}}
\newcommand{\satsymb}{\Vdash}
\newcommand{\basicval}{*}

\newcommand{\model}[1][M]{\modeltext{#1}}
\newcommand{\ikmodel}[1][M]{\modeltext{#1}}
\newcommand{\val}[3][\basicval]{{#2}^{#1}_{#3}}
\newcommand{\ikval}[3][\basicval]{#2^{#1}_{#3}}
\newcommand{\satisfy}[3][\model]{#1 , #2 \satsymb #3}
\newcommand{\iksatisfy}[3][\ikmodel]{#1 , #2 \satsymb #3}

\renewcommand{\models}{\vDash}
\newcommand{\basicmod}{\vDash_\textsf{b}}
\newcommand{\modulmod}{\vDash_\textsf{m}}

\newcommand{\iksat}[1]{\vDash #1}
\newcommand{\ikcsqsem}[2]{#1 \models #2}

\newcommand{\notsatisfy}[3][\model]{#1 , #2 \not\satsymb #3}
\newcommand{\worlds}{W}
\newcommand{\ikworlds}{W}
\newcommand{\ikirel}{\leq}
\newcommand{\ikrel}{R}
\newcommand{\rel}{R}
\newcommand{\irel}{\leq}
\newcommand{\irelop}{\geq}

\newcommand{\cansymb}{c}
\newcommand{\canmodel}[1][M]{\model[#1]^\cansymb}
\newcommand{\canworlds}{\worlds^\cansymb}
\newcommand{\canrel}{\rel^\cansymb}
\newcommand{\canirel}{\irel^\cansymb}

\newcommand{\canbasicval}{\basicval^\cansymb}
\newcommand{\canval}[2]{\val[\canbasicval]{#1}{#2}}

\newcommand{\modmodel}[1][]{\model[#1]_{\BOX}}
\newcommand{\canmodmodel}{\canmodel_{\BOX}}
\newcommand{\BC}{BC}
\newcommand{\FC}{FC}

\newcommand{\primesetsymb}[1]{%
	\ifnum#1=0 \Gamma
    \else\ifnum#1=1 \Delta%
	\else\ifnum#1=2 \Pi%
	\else\ifnum#1=3 \Sigma%
	\else\ifnum#1=4 \Omega%
	\else\Delta%
	\fi\fi\fi\fi\fi
}
\NewDocumentCommand{\primeset}{O{0} O{}}{\primesetsymb{#1}_{#2}}
\NewDocumentCommand{\hashprime}{O{0} O{}}{\primeset[#1][#2]^{\sharp}}
\NewDocumentCommand{\flatprime}{O{0} O{}}{\primeset[#1][#2]^{\flat}}

\makeatletter
\DeclareFontFamily{OMX}{MnSymbolE}{}
\DeclareSymbolFont{MnLargeSymbols}{OMX}{MnSymbolE}{m}{n}
\SetSymbolFont{MnLargeSymbols}{bold}{OMX}{MnSymbolE}{b}{n}
\DeclareFontShape{OMX}{MnSymbolE}{m}{n}{
    <-6>  MnSymbolE5
   <6-7>  MnSymbolE6
   <7-8>  MnSymbolE7
   <8-9>  MnSymbolE8
   <9-10> MnSymbolE9
  <10-12> MnSymbolE10
  <12->   MnSymbolE12
}{}
\DeclareFontShape{OMX}{MnSymbolE}{b}{n}{
    <-6>  MnSymbolE-Bold5
   <6-7>  MnSymbolE-Bold6
   <7-8>  MnSymbolE-Bold7
   <8-9>  MnSymbolE-Bold8
   <9-10> MnSymbolE-Bold9
  <10-12> MnSymbolE-Bold10
  <12->   MnSymbolE-Bold12
}{}

\let\llangle\@undefined
\let\rrangle\@undefined
\DeclareMathDelimiter{\llangle}{\mathopen}%
                     {MnLargeSymbols}{'164}{MnLargeSymbols}{'164}
\DeclareMathDelimiter{\rrangle}{\mathclose}%
                     {MnLargeSymbols}{'171}{MnLargeSymbols}{'171}
\makeatother

\newcommand{\qrelbr}[1]{\llangle #1 \rrangle}
\newcommand{\qrel}[1]{\qrelbr{#1}}
\newcommand{\prelbr}[1]{\llbracket #1 \rrbracket}
\newcommand{\prel}[1]{\prelbr{#1}}

\newcommand{\syp}[1]{\DIA[w]}
\newcommand{\jyb}[1]{\BOX[w]}

\newcommand{\realfunction}{r}
\newcommand{\real}[1]{\realfunction(#1)}

\title{Intuitionistic Justification Logic, Semantically}
\author{Sonia Marin
\qquad\qquad 
Paaras Padhiar 
\qquad\qquad 
Ian Shillito\thanks{The third author has been supported by the UKRI Future Leaders Fellowship, ‘Structure vs Invariants in Proofs’, project reference MR/S035540/1.}
\email{\qquad s.marin@bham.ac.uk \ 
pxp367@student.bham.ac.uk \ 
i.b.p.shillito@bham.ac.uk}
\institute{School of Computer Science\\
University of Birmingham\\
Birmingham, United Kingdom}
}

\newcommand{\titlerunning}{Intuitionistic Justification Logic, Semantically}
\newcommand{\authorrunning}{S. Marin, P. Padhiar \& I. Shillito}

\hypersetup{
  bookmarksnumbered,
  pdftitle    = {\titlerunning},
  pdfauthor   = {\authorrunning},
  pdfsubject  = {EPTCS},               
  pdfkeywords = {Justification logic, Intuitionistic modal logic, Modular models, Realisation theorem} 
}

\begin{document}
\maketitle

\begin{abstract}
Justification logics are explicit versions of modal logic. In the classical setting, this means boxes are refined with explicit proof terms and interact with each other through proof operations. This exercise was extended to intuitionistic modal logic with native diamonds. In this setting, diamonds are refined to satisfier terms and come equipped with additional operations.

Justification logic enjoys a connection to its corresponding modal logic through a realisation theorem. In the classical setting, this is achieved through either proof-theoretic or semantic methodology. So far, intuitionistic justification logic with satisfiers has only been presented syntactically with a proof-theoretic realisation theorem.

We present two classes of semantics for intuitionistic justification logic with soundness and completeness results: basic modular models, which extend possible world semantics for intuitionistic propositional logic; modular models which contain Kripke-style machinery to promote ``backwards compatibility'' to modal logic. Using modular models, we present a realisation theorem to establish a connection between intuitionistic justification logic and its corresponding intuitionistic modal logic.
\end{abstract}

\section{Introduction}
Justification logic makes knowledge explicit by attaching reasons to statements.
It refines modalities in standard modal logic: the usual $\BOX A$, which can mean ``$A$ is \emph{provable} or \emph{known}", is replaced with $\just{t}{A}$ for some proof term $\pt{t}$, interpreted as ``$\pt{t}$ is a \emph{proof} or \emph{evidence} of $A$".

\textbf{Realisation theorem.}
The first proposed justification logic is the Logic of Proofs~($\LP$), introduced by Artemov~\cite{artemov_operational_1995}, which makes the modal logic $\sfour$ explicit.
This is achieved formally via a \emph{realisation theorem} translating each theorem of $\sfour$ into its corresponding theorem in $\LP$ by \emph{realising} every $\BOX$ modality into a suitable proof term.
The realisation theorem was first proved by Artemov using a proof-theoretic (and constructive) methodology~\cite{artemov_explicit_2001},
and later on via semantic techniques by Fitting~\cite{fitting_logic_2005}.

\textbf{Semantics for justification logic.}
Originally, the logic of proofs was interpreted using proofs in Peano Arithmetic~($\PA$)~\cite{artemov_operational_1995}.
Later, simple and flexible semantic models akin to epistemic (possible-world) models were created,
in which justification terms represent evidence rather than
explicit arithmetic proofs~\cite{mkrtychev_models_1997, fitting_logic_2005}.
These models helped apply justification logic in a more general fashion
to the concepts of knowledge and evidence,
better understand what justification terms mean,
and prove important logical properties like decidability~\cite{bucheli_decidability_2011} and complexity bounds~\cite{kuznets_complexity_2000}.

\textbf{Intuitionistic justification logic.}
Since $\LP$ interprets proofs in classical arithmetic, an intuitionistic version of the Logic of Proofs ($\iLP$) was sought to represent proofs in Heyting Arithmetic~($\HA$)~\cite{artemov_basic_2007,dashkov_arithmetical_2011}.
Though $\iLP$ is more than just $\LP$ over intuitionistic logic:
to match the proof interpretation in $\HA$, additional axioms representing admissible rules must be included.
The base logic (without these additional axioms) was shown to be an explicit justification version of intuitionistic modal logic $\isfour$~\cite{artemov_unified_2002}.
Its semantics was studied in~\cite{marti_intutionistic_2016}, where they provide intuitionistic Kripke-style models and prove completeness, and its proof theory in~\cite{hill_analytic_2019}.
Intermediate variants of~$\LP$ were also studied by Pischke~\cite{pischke_intermediate_2023} providing in particular an adaptation of the semantic realisation technique.

\textbf{Intuitionistic diamonds.}
In classical modal logic, $\BOX$ and $\DIA$ are dual, i.e.,
$\DIA A$ can be defined as~$\NOT\BOX\NOT A$,
making the behaviour of $\DIA$ automatically determined by the one of $\BOX$.
However, the weaker negation of intuitionistic logic breaks this duality:
$\DIA$ must be given its own independent semantics and axioms,
as it cannot simply be defined in terms of $\BOX$.
Furthermore, the addition of certain $\DIA$-axioms produces logics which are not conservative over the $\DIA$-free language~\cite{das2023intuitionistic,GroShiClo25,DasGroShi25-blog}, which justifies the interest in intuitionistic variants of modal logics with both modalities~\cite{fischer_servi_modal_1977,simpson_proof_1994}.

\textbf{Diamonds in justification logic.}
In justification logic, even intuitionistic variants, $\DIA$ was never considered as a first-class citizen. 
The first justification logic which makes the $\DIA$ modality explicit was postulated by~\cite{kuznets_justification_2021} as a justification counterpart to \emph{constructive modal logic}~$\CK$~\cite{bellin_extended_2001} with a syntactic realisation result.
This work was then extended to provide a justification counterpart to \emph{intuitionistic modal logic}~$\IK$~\cite{fischer_servi_modal_1977} syntactically in~\cite{marin_justification_2025} with a proof-theoretic procedure based on nested sequents.

\textbf{Contributions.}
We continue this line of work with a semantic exploration of intuitionistic justification logics (with $\DIA$).
We extend the definition of basic intuitionistic modular models and intuitionistic modular models of~\cite{marti_intutionistic_2016}, and establish soundness and completeness results with respect
to these semantics for $\JIK$, the justification counterpart of intuitionistic modal logic $\IK$. 
Building on these results, our investigation is crowned by an adaptation of the semantic realisation procedure in~\cite{pischke_intermediate_2023} to $\JIK$.

\textbf{Outline.}
In Section~\ref{sec:prelim} we introduce the preliminaries on $\IK$
and intuitionistic justification logic $\JIK$, 
notably providing some key properties of the latter logic
our work relies on.
In Section~\ref{sec:basic} we introduce basic intuitionistic modular models which extend~\cite{marti_internalized_2018,pischke_intermediate_2023} to incorporate the $\DIA$ operator, and prove that they are still enough to ensure completeness with $\JIK$.
In Section~\ref{sec:modular} we expand on these to get intuitionistic modular models.
They require more adaptation due to the addition of $\DIA$ on the one hand, 
and to the complexity inherent to the semantics for $\IK$ on the other hand.
Indeed, the confluence conditions traditionally expected in the semantics for $\IK$ are superfluous for justification $\JIK$ (given that persistence already holds in the basic semantics) but they are needed to ensure backward compatibility with $\IK$ in the realisation theorem.
In Section~\ref{sec:real} we are indeed able to prove the realisation theorem, that $\IK$ and $\JIK$ formally correspond, by adapting the semantic method~\cite{fitting_logic_2005,pischke_intermediate_2023} to $\IK$ and the intuitionistic $\DIA$.
We conclude in Section~\ref{sec:concl} with some discussions on avenues for taking these ideas further.

\section{Preliminaries}\label{sec:prelim}

\subsection{Intuitionistic Modal Logic}
Let $\Prop$ be a countable set of propositional variables.
The language of modal logic~$\langmod$ is defined by
$$
    G \BNFequals p\in\Prop \mid \BOT \mid G \AND G \mid G \OR G \mid G \IMP G \mid \BOX G \mid \DIA G
$$

Let $\IPLax$ be a finite set of axioms for $\IPL$, e.g.~\cite{troelstra_basic_2000}.
We add a superscript $I$ to denote the set of axiom instances of a 
given set of axioms,
e.g.~$\axinst{\IPLax}$ for the set of axiom instances of $\IPLax$.
We define $\IKax$ to be the set of axioms extending $\IPLax$ with the \emph{modal} axioms presented below~\cite{plotkin_framework_1986}.
\begin{center}
$
    \begin{array}{rcl}
        \kax[1] & : & \BOX (G \IMP H) \IMP (\BOX G \IMP \BOX H) \\
        \kax[2] & : & \BOX (G \IMP H) \IMP (\DIA G \IMP \DIA H) \\
    \end{array}
    \quad
    \begin{array}{rcl}
        \kax[3] & : & \DIA (G \OR H) \IMP (\DIA G \OR \DIA H) \\
        \kax[4] & : & (\DIA G \IMP \BOX H) \IMP \BOX (G \IMP H) \\
    \end{array}
    \quad
    \begin{array}{rcl}
        \kax[5] & : & \DIA \BOT \IMP \BOT
    \end{array}
$
\end{center}

    \begin{definition}\label{def:ax-sys-mod}
    The logic $\IK$ is the set of consecutions $\proves[\Gamma]{}{G}$,
    where $\Gamma=\set{G_1,\dots,G_n}$ is a finite set of formulas,
    derivable in the system described by the following rules.
    \begin{center}
    \begin{tabular}{c@{\hspace{1.2cm}}c @{\hspace{1.2cm}}c @{\hspace{1.2cm}}c}
    \hypertarget{rule:ax}{
    $
    \inferLineSkip=3pt
    \infer[\axrule]{\proves[\Gamma]{}{G}}{G\in\axinst{\IKax}}
    $}
    &
    \hypertarget{rule:id}{
    $
    \inferLineSkip=3pt
    \infer[\idrule]{\proves[\Gamma]{}{G_i}}{1\leq i \leq n}
    $}
    &
    \hypertarget{rule:nec}{
    $
    \inferLineSkip=3pt
    \infer[\nec]{\proves[\Gamma]{}{\BOX G}}{\proves[\emptyset]{}{G}}
    $}
    &
    \hypertarget{rule:MP}{$
    \inferLineSkip=3pt
    \infer[\modus]{\proves[\Gamma]{}{H}}{
    	\proves[\Gamma]{}{G}
    	&
    	\proves[\Gamma]{}{G\IMP H}}
    $}
    \end{tabular}
    \end{center}
    \vspace{0.2cm}
    
    If the consecution $\proves[\Gamma]{}{G}$ is provable in the system for $\IK$, we write $\proves[\Gamma]{\IK}{G}$.
    We omit the curly brackets when clear from context and write $\proves[G_1, \dots, G_n]{\IK}{G}$ for $\proves[\set{G_1, \dots, G_n}]{\IK}{G}$.
    For a potentially infinite set of formulas $\primeset$, we write $\proves[\primeset]{\IK}{G}$ if there are formulas $G_1, \dots, G_n \in \Gamma$ such that $\proves[G_1, \dots, G_n]{\IK}{G}$.
    We write $\proves[\primeset, H]{\IK}{G}$ for $\proves[\primeset \UNI \set{H}]{\IK}{G}$, and $\proves[]{\IK}{G}$ for $\proves[\emptyset]{\IK}{G}$.
    \end{definition}

    \begin{remark}
    The fact that the premise of the $\nec$ rule requires an empty context 
    entails that we capture the \emph{local} modal logic $\IK$, understood
    as a consequence relation.
    A benefit of this rule is that the deduction theorem holds~\cite{hakli_does_2012}, a result implicitly used 
    in some of the results below about $\IK$.
    \end{remark}
With the syntax and axiomatic system of $\IK$ defined,
we turn to its (bi)relational semantics~\cite{fischer_servi_modal_1977,fischer_servi_semantics_1980}.
\begin{definition}[Frames]\label{def:frames}
    A \emph{(birelational) frame} is a tuple $\tuple{\ikworlds, \ikirel, \ikrel}$, where $W$ is a non-empty set of worlds, $\ikirel$ is a preorder on~$\ikworlds$ and $\ikrel$ is a binary relation on~$\ikworlds$ such that:
    \begin{itemize}
        \item[($\BC$)] for all worlds $w, v, v'$ with $w \ikrel v \ikirel v'$, then there exists a world $w'$ with $w \ikirel w' \ikrel v'$;

        \item[($\FC$)] for all worlds $w, w', v$ with $w \ikirel w'$ and $w \ikrel v$, there exists a world $v'$ with $w' \ikrel v'$ and $v \ikirel v'$.
    \end{itemize}
    \begin{center}
    \begin{tabular}{c @{\hspace{3cm}} c}
    \begin{tikzpicture}[thick, every node/.style={scale=0.9}]
		\tikzstyle{access}=[->]
		\tikzstyle{future}=[->]
		
		\node (3) at (5,3) [] {$\BC$};
		
		\node (3) at (4,2) {$w$};
		\node (4) at (6,2) {$v$};
		\node (10) at (4,4){$w'$};
		\node (11) at (6,4) {$v'$};
		
		\draw[access] (3) edge[below] node {$R$} (4);
		\draw[access,dotted] (10) -- (11);

		\draw[future,dotted] (3) -- (10);
		\draw[future] (4) edge[right] node {$\le$}  (11);
	\end{tikzpicture}
    & 
    \begin{tikzpicture}[thick, every node/.style={scale=0.9}]
	\tikzstyle{access}=[->]
	\tikzstyle{future}=[->]
	
	\node[] (3) at (5,3) [] {$\FC$};
	
	\node (3) at (4,2) {$w$};
	\node (4) at (6,2) {$v$};
	\node (10) at (4,4) {$w'$};
	\node (11) at (6,4) {$v'$};
	
	\draw[access] (3) edge[below] node {$R$} (4);
	\draw[access,dotted] (10) -- (11);
	
	\draw[future] (3) edge[left] node {$\le$} (10);
	\draw[future,dotted] (4) -- (11);
\end{tikzpicture}
    \end{tabular}
    \end{center}
\noindent Note that $\BC$ stands for \emph{B}ackward \emph{C}onfluence, while
$\FC$ stands for \emph{F}orward \emph{C}onfluence.
\end{definition}

\begin{definition}[Models]
    A \emph{(birelational) model} for $\IK$, 
    is a tuple $\ikmodel = \tuple{\ikworlds, \ikirel, \ikrel, \basicval}$ with $\tuple{\ikworlds, \ikirel, \ikrel}$ a frame equipped with a map $\function{\basicval}{\Prop \cartprod \ikworlds}{\set{0, 1}}$ 
    such that $w \ikirel v$ and $\basicval(p,w)=1$ entail~$\basicval(p,v)=1$.
    We henceforth use the notation $\ikval{p}{w}$ for $\basicval(p,w)$.
    For a model $\ikmodel$, a point $w\in W$ and a formula $G$,
    the \emph{truth of $G$ at $w$ in $\ikmodel$} is recursively defined on the structure of $G$:
    $$
        \begin{array}{lcl}
            \iksatisfy{w}{p} & \text{iff} & \ikval{p}{w} = 1 \\
            \iksatisfy{w}{\BOT} & & \text{never} \\
            \iksatisfy{w}{G \AND H} & \text{iff} & \text{$\iksatisfy{w}{G}$ and $\iksatisfy{w}{H}$} \\
            \iksatisfy{w}{G \OR H} & \text{iff} & \text{$\iksatisfy{w}{G}$ or $\iksatisfy{w}{H}$} \\
            \iksatisfy{w}{G \IMP H} & \text{iff} & \text{for all worlds $w'$ with $w \ikirel w'$, if $\iksatisfy{w'}{G}$ then $\iksatisfy{w'}{H}$} \\
            \iksatisfy{w}{\BOX G} & \text{iff} & \text{for all worlds $w', v'$ with $w \ikirel w' \ikrel v'$, $\iksatisfy{v'}{G}$} \\
            \iksatisfy{w}{\DIA G} & \text{iff} & \text{there exists world $v$ with $w \ikrel v$, $\iksatisfy{v}{G}$}
        \end{array}
    $$
    We write $\iksatisfy{w}{\Gamma}$ if $\iksatisfy{w}{G}$ for all $G \in \Gamma$.
    We write $\iksat{G}$ if $\iksatisfy{w}{G}$ for any~$\ikmodel$ and~$w\in\ikworlds$.
\end{definition}

The following relates the semantic and axiomatic definitions of $\IK$.
\begin{theorem}[Soundness and Completeness~\cite{fischer_servi_axiomatizations_1984,plotkin_framework_1986,simpson_proof_1994}]
    $\proves{\IK}{G} \iff \ikcsqsem{}{G}$. 
\end{theorem}

Later in the paper, we make use of the following consequence of soundness.
\begin{corollary}
    $\IK$ is consistent.
\end{corollary}
\subsection{Intuitionistic Justification Logic}
    The sets of \emph{proof terms} $\prftm$ and of \emph{satisfier terms} $\satterms$ are defined by mutual induction 
    as follows
    \begin{center}
	$
		\tm t \BNFequals \pconst \in\prfconst \mid \pvar \in\prfvars \mid \jappl{t}{t} \mid \jsum{t}{t} \mid \jbang{t} \mid \jupdt{m}{t}
	       \qquad
		\tm m \BNFequals \svar \in\satvars \mid \suni{m}{m} \mid \sappl{t}{m}
	$
    \end{center}
    where proof constants~$\prfconst$,
    proof variables~$\prfvars$ and
    satisfier variables $\satvars$ are countable sets.
    The \emph{language} of justification logic~$\langjust$ reuses the set $\Prop$ as shown in the following grammar
	$$
		A \BNFequals p\in\Prop \mid \BOT \mid A \AND A \mid A \OR A \mid A \IMP A \mid \just{t}{A} \mid \sat{m}{A}
	$$
	where $\tm t \in\prftm$ and $\tm m \in\satterms$.
    Before defining logics over $\langjust$, we provide some intuitions on the intended meaning of the operators over proof and satisfier terms.

    In justification logic, a formula of the shape $\just{t}{A}$ can be read as 
    $\term t$ is a \emph{proof} of~A.
    Proof terms can hence be thought of as capturing \emph{global reasoning}, i.e.~they assert the \emph{validity} of statements with respect to any model.
    Dually, satisfier terms relate to \emph{local reasoning}: $\sat{m}{A}$ could be read as $\tm m$ is a \emph{model} where $A$ is satisfied, implying the \emph{consistency} of $A$.
    The operations \emph{proof sum}~$\jsum{}{}$, \emph{application}~$\jappl{}{}$ and \emph{proof checker}~$\jbang{}$ are standard justification operations relating to common proof manipulations~\cite{artemov_operational_1995,artemov_explicit_2001}.
    The operation on satisfiers of \emph{disjoint union}~$\suni{}{}$ is quite naturally the local counterpart to the proof sum. 
    The \emph{propagation} operation~$\sappl{}{}$ 
    combines global and local reasoning by taking as arguments a proof term and a satisfier term~\cite{kuznets_justification_2021}.
    For example, the locality of $A$ in $\sat{m}{A}$
    can be exploited via the global holding of $A \IMPLIES (A \OR B)$
    in $\just{\tm t}{(A \IMPLIES (A \OR B))}$ to locally establish $(A \OR B)$
    as $\sat{\sappl{t}{m}}{(A \OR B)}$.
    Finally, the \emph{local update} operation~$\jupdt{}{}$ expresses that if local reasoning implies global reasoning, global information can be updated using local information~\cite{marin_justification_2025}.
    For example, if the satisfier $\tm m$ locally reasons about $A$ and implies the proof term $\tm t$ globally reasons about $B$, then this connection is \emph{updated} into global reasoning by making  the proof term $\jupdt{m}{t}$ globally reason about $A \IMP B$.
    While the interpretations of operations on proof terms are formally backed up by the completeness proof with respect to Peano arithmetic in the context of classical Logic of Proofs~\cite{artemov_explicit_2001}, in the intuitionistic context, these merely constitute guidelines to help us reason about the operations at a high-level and would require to be studied within a meta-theory such as Heyting arithmetic or related to be given a proper formal understanding.
    The interpretation we gave to our operators is reflected in the set of axioms $\BaseAx$, which extends $\IPLax$ with the additional \emph{justification} axioms
    presented in the two left columns of Figure~\ref{fig:justif-axioms}.

	\begin{figure}[t]
    $$
				\begin{array}{r@{\ :\ }l}
					\jkax[1] & \just{s}{(A \IMP B)} \IMP (\just{t}{A} \IMP \just{\jappl{s}{t}}{B})
					\\
					\jkax[2] & \just{s}{(A \IMP B)} \IMP (\sat{m}{A} \IMP \sat{\sappl{s}{m}}{B})
					\\
					\jkax[3] & \sat{m}{(A \OR B)} \IMP (\sat{m}{A} \OR \sat{m}{B})
					\\
					\jkax[4] & (\sat{m}{A} \IMP \just{t}{B}) \IMP \just{\jupdt{m}{t}}{(A \IMP B)}
					\\
					\jkax[5] & \sat{m}{\BOT} \IMP \BOT
				\end{array}
				\quad\quad
				\begin{array}{r@{\ :\ }l}
					\jsumaxl & \hypertarget{ax:jsuml}{\just{s}{A} \IMP \just{\jsum{s}{t}}{A}}
					\\
					\jsumaxr & \hypertarget{ax:jsumr}{\just{t}{A} \IMP \just{\jsum{s}{t}}{A}}
					\\
					\juniaxl & \hypertarget{ax:junil}{\sat{m}{A} \IMP \sat{\suni{m}{n}}{A}}
					\\
					\juniaxr & \hypertarget{ax:junir}{\sat{n}{A} \IMP \sat{\suni{m}{n}}{A}}
				\end{array}
        $$
    \caption{Intuitionistic justification axioms and rule}
    \label{fig:justif-axioms}
	\end{figure}

    While traditionally a set of axioms like $\BaseAx$ would generate a
    \emph{unique} logic, justification logic allows for the definition of
    myriad logics via \emph{constant specifications}.

	\begin{definition}[Constant specification]
		A constant specification for~$\BaseAx$ is any subset
		$
		\CS \sset \set[\just{c}{A}]{\text{$\tm c \in \prfconst$ and $A\in\BaseAx^I$}}
		$.
        We call a constant specification~$\CS$ \emph{axiomatically appropriate} if for each axiom instance~$A$ of~$\BaseAx$, there exists a proof constant~$\tm c$ such that $\just{c}{A} \in \CS$.
        We call an axiomatically appropriate constant specification~$\CS$ \emph{schematic}
        if for any instance $A$ of an axiom in~$\BaseAx^I$ we have that $\just{c}{A} \in \CS$ entails $\just{c}{A'} \in \CS$ for any instance $A'$ of the same axiom.
	\end{definition}

    The parametricity in constant specifications of axiomatic systems for justification logics 
    is illustrated by a characteristic rule of these calculi.

    \begin{definition}\label{def:ax-sys-just}
    Given
    a constant specification $\CS$ for $\BaseAx$, we define the logic $\JIKCS$ as the set of consecutions $\proves[\Gamma]{}{A}$ provable by the system described by the rules from Definition~\ref{def:ax-sys-mod} (over $\langjust$), where:
    \begin{itemize}
    \item $\IKax$ is replaced by $\BaseAx$ and
    \item $\nec$ is replaced by the rule $\can$ presented below.
    $$
    \hypertarget{rule:can}{
                \inferLineSkip=3pt
                \infer[\can]{\Gamma\vdash \just{\jbang[k]{c}}{\dots\just{\jbang{c}}{\just{\tm c}{A}}}}
                {
                    \just{\tm c}{A} \in \CS
                    &
                    k \geq 0}}
    $$
    \end{itemize}
    We adopt the conventions of Definition~\ref{def:ax-sys-mod} with $\proves{\IK}{}$ replaced by $\proves{\JIKCS}{}$.
    \end{definition}

    \begin{remark}    
        The $\jbang{}$ operator is not necessary to give a justification counterpart to non-transitive logics and we could opt for a formulation where
        $
		\CS \sset \set[\just{c_n}{\just{\dots}{\just{c_1}{A}}}]{\text{$\tm{c_1}, \dots, \tm{c_n} \in \prfconst$ and $A\in\BaseAx^I$}}
		$~\cite{artemov_explicit_2001}.
        We prefer to use the explicit $\jbang{}$ operator for a more concise presentation of $\CS$ where proof constants can really be thought of as a \emph{witness} to an axiom, and not to more complex formula of the form~${\just{c}{A}}$.
    \end{remark}

    The following is an example of a proof in $\JIKCS$.
    \begin{example}\label{ex:justFSaxiomproof}
        We show that from $\proves{\JIKCS}{\just{t}{(A \IMP (A \IMP B) \IMP B)}}~(1)$ one can infer 
        $\proves{\JIKCS}{\sat{m}{(A \IMP B)} \IMP \just{s}{A} \IMP \sat{\sappl{(\jappl{t}{s})}{m}}{B}}$
        for any proof term~$\tm s$ and satisfier term~$\tm m$.
        First, we instantiate $\jkax[1]$ and $\jkax[2]$ as follows:
        $\proves{\JIKCS}{\just{t}{(A \IMP (A \IMP B) \IMP B)} \IMP \just{s}{A} \IMP \just{\jappl{t}{s}}{((A \IMP B) \IMP B))}}~(2)$
        and
        $\proves{\JIKCS}{\just{\jappl{t}{s}}{((A \IMP B) \IMP B))} \IMP \sat{m}{(A \IMP B)} \IMP  \sat{\sappl{(\jappl{t}{s})}{m}}{B}}~(3)$.
        Second, we apply $\modus$ to (1) and (2) to obtain $\proves{\JIKCS}{\just{s}{A} \IMP \just{\jappl{t}{s}}{((A \IMP B) \IMP B))}}~(4)$.
        Via transitivity of~$\IMP$ on (3) and (4), we get $\proves{\JIKCS}{\just{s}{A} \IMP \sat{m}{(A \IMP B)} \IMP  \sat{\sappl{(\jappl{t}{s})}{m}}{B}}~(5)$.
        Then by propositional reasoning on (5), we finally infer~$\proves{\JIKCS}{\sat{m}{(A \IMP B)} \IMP \just{s}{A} \IMP \sat{\sappl{(\jappl{t}{s})}{m}}{B}}$.
        This is the justification version of the
        theorem $\DIA(A \IMP B) \IMP \BOX A \IMP \DIA B$ in $\IK$.
    \end{example}

\subsection{Properties of Justification Logic}
In this section, we recall key definitions and properties of justification logic needed
throughout the paper.

	\begin{theorem}[Deduction Theorem]\label{thm:deduction}
		Let $\Gamma\cup\set{A,B}\subseteq\langjust$.
		Then:
		$\proves[\Gamma,A]{\JIKCS}{B} \iff \proves[\Gamma]{\JIKCS}{A \IMP B}$
	\end{theorem}
	\begin{proof}
    This is a standard argument~\cite{kuznets_logics_2019}:
    from right to left we simply use $\modus$ and monotonicity on the left
    of provability, and in the other direction we proceed by induction on the 
    given proof of $\proves[\Gamma,A]{\JIKCS}{B}$.
	\end{proof}

    Next, we show that $\JIKCS$ is closed under a notion of substitution modifying proof and satisfier terms.

	\begin{definition}[Justification substitution]\label{def:substitution}
		A \emph{(justification) substitution} is a pair of maps~$\function{\sub_t}{\prfvars}{\prfterms}$ and~$\function{\sub_s}{\satvars}{\satterms}$.
        For simplicity, we merge the two maps and write $\sub$ as a justification substitution.
        We define as expected the applications $\subst{\tm t}$, $\subst{\tm m}$ and
        $\subst{A}$ of $\sigma$ to, respectively, the proof term $\tm t$, the satisfier term $\tm m$ and the formula $A$.
        We write $\subst{\Gamma}$ to designate the set $\{\subst{A}\mid A\in\Gamma\}$.
	\end{definition}
        \begin{lemma}[Substitution Lemma]\label{lem:substitution}
        Let $\CS$ be a
        schematic constant specification, 
        $\Gamma\cup\set{A}\subseteq\langjust$,
        and $\sub$ a substitution.
        Then:
        $\proves[\Gamma]{\JIKCS}{A} \implies \proves[\subst{\Gamma}]{\JIKCS}{\subst{A}}$
    \end{lemma}

    \begin{proof}
        The proof is similar to~\cite[Lemma 2.27]{kuznets_logics_2019},
        and goes by induction on the structure of the proof of $\proves[\Gamma]{\JIKCS}{A}$.
        The critical cases are when the last rule applied is $\axrule$ or $\can$.
        In the former case, it suffices to notice that if $A$ is an instance of an axiom, then so is $\subst{A}$.
        In the latter case, we make use of the fact that $\CS$ is
        schematic:
        with $\just{\pconst}{A}\in\CS$ and $\just{\jbang[n]{\pconst}}{\dots{\just{\pconst}{A}}}$ appearing in our consecution,
        we obtain a proof of $\proves[\subst{\Gamma}]{}{\subst{(\just{\jbang[n]{\pconst}}{\dots{\just{\pconst}{A}}})}}$ via $\can$ by noticing that 
        $\subst{(\just{\jbang[n]{\pconst}}{\dots{\just{\pconst}{A}}})} = \just{\jbang[n]{\pconst}}{\dots{\just{\pconst}{(\subst{A})}}}$
        and that $\just{\pconst}{\subst{A}}\in\CS$ as $\CS$ is schematic.
    \end{proof}

    The next lemma expresses an interesting feature of justification logic: 
    its (object) language allows for the reflection of its own (meta) proofs
    within formulas.
    On top of its conceptual interest, we leverage this lemma in many places 
    throughout the paper.

    \begin{lemma}[Lifting Lemma]\label{lem:lifting}
        Let $\CS$ be axiomatically appropriate.
        Let $k \in \nat$,
        $\set{A_1, \dots, A_k, B, A}\subseteq\langjust$,
        $\set{\tm t_1, \dots, \tm t_k} \subseteq \prftm$ and $\tm m \in \satterms$.
        Then
        \begin{enumerate}
            \item If $\proves[A_1, \dots, A_k]{\JIKCS}{A}$ then there exists a proof term $\tm t$ s.t.~$\proves[\just{\tm t_1}{A_1}, \dots, \just{\tm t_k}{A_k}]{\JIKCS}{\just{\tm t}{A}}$.
            
            \item If $\proves[A_1, \dots, A_k, B]{\JIKCS}{A}$ then there exists a satisfier term $\tm n$ s.t.~$\proves[\just{\tm t_1}{A_1}, \dots, \just{\tm t_k}{A_k}, \sat{\tm m}{B}]{\JIKCS}{\sat{\tm n}{A}}$.
        \end{enumerate}
    \end{lemma}

    \begin{proof}
        As in~\cite{marin_justification_2025}, we proceed by induction on the structure of the proof of 
        $\proves[A_1, \dots, A_k]{\JIKCS}{A}$, thereby inspecting the last rule applied.
        Our requirement that $\CS$ is axiomatically appropriate shows in the base case when $A\in\axinst{\BaseAx}$:
        the existence of a constant $\term c$ such that $\just{c}{A} \in \CS$ is ensured, thereby giving us $\proves[\just{\tm t_1}{A_1}, \dots, \just{\tm t_k}{A_k}]{\JIKCS}{\just{c}{A}}$ via $\can$.
        We expand on the case for $\can$, as the remaining cases are standard.
        In this case $A = \just{\jbang[l]{c}}{\just{\dots}{\just{\jbang{c}}{\just{\tm c}{B}}}}$ 
            for some proof constant $\tm c \in \prfconst$ and formula $B$ with $\just{\tm c}{B} \in \CS$.
            Instantiating $\can$ one step further, we get
            $\proves[\just{\tm t_1}{A_1}, \dots, \just{\tm t_k}{A_k}]{\JIKCS}{\just{\jbang[l+1]{c}}{\just{\jbang[l]{c}}{\just{\dots}{\just{\jbang{c}}{\just{\tm c}{B}}}}}}$
            and hence
            $\proves[\just{\tm t_1}{A_1}, \dots, \just{\tm t_k}{A_k}]{\JIKCS}{\just{\jbang[l+1]{c}}{A}}$.
    \end{proof}

    The proof of the above lemma informs us that the terms $\tm t$ and $\tm n$ we build 
    are parametric in $\tm t_1,\dots,\tm t_k$ and $\tm m$. 
    As a consequence, when $k=0$ the first part of the lemma forces $\tm t$ to be a \emph{ground} term, i.e.~a proof term which contains no proof or satisfier variable.

    The technical lemma below generalises the axioms governing the operators $\mathbin{+}$ and $\mathbin{\sqcup}$,
    a useful tool in the completeness proof we provide later on.

    \begin{lemma}\label{prop:orterms}
        Let $\CS$ be axiomatically appropriate.
        Let $i,j \in \nat$,
        $\set{A_1, \dots, A_i, B_1, \dots, B_j}\subseteq\langjust$,
        $\set{\tm t_1, \dots, \tm t_i}$ $\subseteq \prftm$ and
        $\set{\tm m_1, \dots, \tm m_j}\sset \satterms$.
        Then there exists a proof term $\tm t$ and a satisfier term $\tm m$ such that:
        \begin{enumerate}
            \item $\proves{\JIKCS}{(\just{\tm t_1}{A_1} \OR \dots \OR \just{\tm t_i}{A_i}) \IMP \just{\tm t}{(A_1 \OR \dots \OR A_i)}}$
            \item $\proves{\JIKCS}{(\sat{\tm m_1}{A_1} \OR \dots \OR \sat{\tm m_j}{A_j}) \IMP \sat{\tm m}{(A_1 \OR \dots \OR A_j)}}$
        \end{enumerate}
    \end{lemma}
    \begin{proof}
    By the Lifting Lemma~\ref{lem:lifting} and using the axioms $\jsumaxl$ and $\jsumaxr$ repetitively for the first statement, or $\juniaxl$ and $\juniaxr$ for the second statement.
    \end{proof}

\subsection{From Modal to Justification -- and back}
    We now proceed to formalise the connection between $\JIKCS$ and $\IK$:
    in one direction, through translating their respective languages;
    in the other, by using \emph{realisation} -- this is the overarching goal of our paper.

	\begin{definition}[Forgetful projection]\label{def:forgetful}
		The \emph{forgetful projection} is a map $\function{\forget{(\cdot)}}{\langjust}{\langmod}$
		inductively defined below, where $\ast \in \set{\AND, \OR, \IMP}$.

        \begin{itemize*}[itemjoin={\qquad}]
        \item $\forget{\BOT} \colonequals \BOT$
        \item $\forget{p} \colonequals p$
        \item $\forget{(A \ast B)} \colonequals (\forget{A} \ast \forget{B})$
        \item $\forget{(\just{\tm t}{A})} \colonequals \BOX \forget{A}$
        \item $\forget{(\sat{\tm m}{A})} \colonequals \DIA \forget{A}$
        \end{itemize*}
	\end{definition}

	\begin{theorem}\label{thm:forgetful}
		Let~$A \in \langjust$.
		If $\proves[\Gamma]{\JIKCS}{A}$ then $\proves[\forget{\Gamma}]{\IK}{\forget{A}}$.
	\end{theorem}
	\begin{proof}
		This follows from the fact that the forgetful projection on axioms of~$\JIKCS$ and conclusions of the $\can$~rule are theorems of~$\IK$.
	\end{proof}
	
	\begin{corollary}\label{cor:consistent}
		$\JIKCS$ is consistent.
	\end{corollary}
	\begin{proof}
		Suppose $\proves{\JIKCS}{\BOT}$.
		Then by Theorem~\ref{thm:forgetful}, $\proves{\IK}{\BOT}$ which is a contradiction.
	\end{proof}

    The reverse direction is expressed via realisation maps.
    
    \begin{definition}[Realisation map]
    A \emph{realisation map} is a function $\function{\realfunction}{\langmod}{\langjust}$ such that $\forget{\real G} = G$ for each $G \in \langmod$.
    \end{definition}

    In the remaining of the paper, we set ourselves to prove the next theorem.

    \begin{theorem}[Realisation Theorem]\label{thm:realisation}
        For schematic $\CS$, there is a realisation map $r$ such that:
         $$\forall G\in\langmod.\;\;\;\proves{\IK}{G} \implies \proves{\JIKCS}{\real{G}}.$$
    \end{theorem}
This \emph{realisation} map 
is precisely what 
embeds the modal logic $\IK$ into its corresponding justification logic $\JIKCS$.
The condition $\forget{\real{G}} = G$ ensures that each $\BOX$ (and each $\DIA$) occurring in $G$ is replaced by exactly one proof term (and satisfier term, respectively).
In other words, a realisation map takes a modal formula $G$ into a justification formula with the same formula tree.

\section{Basic Modular Models}\label{sec:basic}
Basic modular models were first introduced by Mkrtychev~\cite{mkrtychev_models_1997} in a classical setting.
These models can be seen as an extension of valuations for classical propositional logic,
which assign a truth value to each proposition,
with an additional interpretation on proof terms.

In the intuitionistic setting, Marti and Studer~\cite{marti_intutionistic_2016}
build basic modular models for an intuitionistic justification logic by instead enhancing the relational semantics for intuitionistic propositional logic.
In this section, we extend these models with an interpretation on satisfier terms.

\subsection{Definition and Soundness}
For the definition of basic models for $\JIKCS$,
expanding upon the basic models given in~\cite{marti_intutionistic_2016},
we need the following operation on sets of formulas
$\primeset \jappl{}{} \primeset[1] := \set[B \in \langjust]{\text{$A \IMP B \in \primeset\text{ and }A \in \primeset[1]$}}$.
\begin{definition}[Basic model]\label{def:basicmodmodel}
		A \emph{basic (intuitionistic modular) model}
        is a tuple~$\model[B] = (\worlds, \irel,\basicval_\terms,\basicval_\Prop)$, where $\worlds$ is a non-empty set of worlds, $\irel$ is a pre-order on~$\worlds$, 
        $\function{\basicval_\Prop}{\Prop \cartprod \worlds}{\set{0,1}}$, and
		$\function{\basicval_\terms}{(\prfterms \UNI \satterms) \cartprod \worlds}{\pset{\langjust}}$.
        
        We abuse notation and combine $\basicval_\terms$ and $\basicval_\Prop$
        into a single function $\basicval$, making basic modular models tuples of the shape $(\worlds, \irel,\basicval)$. 
        We also write $\val{\term s}{w}$ for $\basicval(\term s,w)$.
        In a basic modular model, the function $\basicval$ satisfies the following conditions on proof and satisfier terms:

        \begin{center}
        \begin{tabular}{r @{\hspace{0.2cm}} l @{\hspace{2cm}} r @{\hspace{0.2cm}} l}
         1. & $\jappl{\val{\term s}{\mathnormal{w}}}{\val{\term t}{\mathnormal{w}}} \sset \val{(\jappl{s}{t})}{w}$. &   
         5. & $A \in \val{\term t}{w}$ for any conclusion $\just{t}{A}$ of $\can$. \\
         2. & $\jappl{\val{\tm s}{\mathnormal{w}}}{\val{\tm m}{\mathnormal{w}}} \sset \val{(\sappl{s}{m})}{\mathnormal{w}}$. &   
         6. & If $A \OR B \in \val{\tm m}{w}$ then $A \in \val{\tm m}{w}$ or $B \in \val{\tm m}{w}$. \\
         3. & $\val{\term s}{w} \UNI \val{\term t}{w} \sset \val{(\jsum{s}{t})}{w}$. &   
         7. & If for all $v \irelop w$, $A \notin \val{\tm m}{v}$ or $B \in \val{\tm t}{v}$, then $A \IMP B \in\val{(\jupdt{\tm m}{t})}{w}$. \\
         4. & $\val{\term m}{w} \UNI \val{\term n}{w} \sset \val{(\suni{m}{n})}{w}$. &   
         8. & $\BOT \notin \val{\tm m}{w}$. \\
        \end{tabular}
        \end{center}
        \end{definition}

        \begin{definition}[Truth in a basic model]\label{def:trubasmod}
        Given a basic model $\model[B]$, a point $w\in W$ and a formula $A$, we define
        the \emph{truth of $A$ at $w$ in $\model[B]$} recursively on the structure of $A$:
        $$
		\begin{array}{rcl}
			\satisfy[{\model[B]}]{w}{p} & \text{iff} & \val{p}{w} = 1 
            \\
			\satisfy[{\model[B]}]{w}{\BOT} & & \text{never} 
            \\ 
            \satisfy[{\model[B]}]{w}{A \AND B} & \text{iff} & \text{$\satisfy[{\model[B]}]{w}{A}$ and $\satisfy[{\model[B]}]{w}{B}$} 
            \\
            \satisfy[{\model[B]}]{w}{A \OR B} & \text{iff} & \text{$\satisfy[{\model[B]}]{w}{A}$ or $\satisfy[{\model[B]}]{w}{B}$} 
            \\
			\satisfy[{\model[B]}]{w}{A \IMP B} & \text{iff} & \text{$\forall v \irelop w$ if $\satisfy[{\model[B]}]{v}{A}$ then $\satisfy[{\model[B]}]{v}{B}$} 
            \\
			\satisfy[{\model[B]}]{w}{\just{t}{A}} & \text{iff} & A \in \val{\term t}{w} 
            \\
			\satisfy[{\model[B]}]{w}{\sat{m}{A}} & \text{iff} & A \in \val{\term m}{w} 
            \\
		\end{array}
		$$        
        We write $\basicmod A$ if for any basic model $\model[B]$ and $w \in \worlds$ we have $\satisfy[{\model[B]}]{w}{A}$.
	\end{definition}

    The monotonicity conditions imposed on $\basicval$ port to the truth of formulas.

	\begin{lemma}[Monotonicity Lemma]\label{lem:basmon}
		Let $\model[B] = (\worlds, \irel, \basicval)$ be a basic modular model, %
		$w$ and $v\in\worlds$ such that $w\leq v$,
		and $A$ be a formula.
		Then: $\satisfy[{\model[B]}]{w}{A} \implies \satisfy[{\model[B]}]{v}{A}$.
	\end{lemma}
	
	\begin{proof}
		We proceed by induction on~$A$. The base case $A = p$ follows by the \textsf{(M1)} condition in Definition~\ref{def:basicmodmodel}.
        The cases $A = B \AND C, B \OR C, B \IMP C$ follow from a standard argument.
        When $A = \just{\tm t}{B}$, we get $B \in \val{\tm t}{w}$ by Definition~\ref{def:trubasmod}.
	    Then, using \textsf{(M2)} of Definition~\ref{def:basicmodmodel} we obtain $B \in \val{\tm t}{v}$ hence $\satisfy[{\model[B]}]{v}{\just{\tm t}{B}}$.
		The case where $A = \sat{\tm m}{B}$ is similar.
    \end{proof}

    The soundness of $\JIKCS$ with respect to its semantics on basic models can now be established.

	\begin{theorem}[Soundness]\label{thm:soundbasicmodularmodel}
		For any formula $A$: $\proves{\JIKCS}{A} \implies \basicmod {A}$.
	\end{theorem}

    \begin{proof}
		By induction on the proof of $A$ in $\JIKCS$.
        Fix a basic model ${\model[B]}$ and a world $w \in \worlds$.
		  We show the validity of the new axioms and rules, and refer to~\cite{marti_intutionistic_2016} 
        for the remaining ones.

		\begin{itemize}

			\item $\jkax[2]: \just{s}{(A \IMP B)} \IMP \sat{m}{A} \IMP \sat{\sappl{s}{m}}{B}$.
			Let $v \irelop w$ with $\satisfy[{\model[B]}]{v}{\just{s}{(A \IMP B)}}$ i.e.~$A \IMP B \in \val{\tm s}{v}$.
			Let $u \irelop v$ with $\satisfy[{\model[B]}]{u}{\sat{m}{A}}$, i.e. $A \in \val{\tm m}{u}$.
			By the monotonicity property, $A \IMP B \in
            \val{\tm s}{u}$.
			So $B \in \jappl{\val{\tm s}{u}}{\val{\tm m}{u}} \sset \val{(\sappl{s}{m})}{u}$ and hence $\satisfy[{\model[B]}]{u}{\sat{\sappl{s}{m}}{B}}$.
			Therefore by definition we have $\satisfy[{\model[B]}]{v}{\sat{m}{A} \IMP \just{\sappl{s}{m}}{B}}$ and hence $\satisfy[{\model[B]}]{w}{\just{s}{(A \IMP B)} \IMP \sat{m}{A} \IMP \just{\sappl{s}{m}}{B}}$.
			
			\item $\jkax[3]: \sat{m}{(A \OR B)} \IMP (\sat{m}{A} \OR \sat{m}{B})$.
			Let $v \irelop w$ with $\satisfy[{\model[B]}]{v}{\sat{m}{(A \OR B)}}$ i.e.~$A \OR B \in \val{\tm m}{v}$. 
            By property 6 of Definition~\ref{def:basicmodmodel},
            $A \in \val{\tm m}{v}$ or $B \in \val{\tm m}{v}$.
			So $\satisfy[{\model[B]}]{v}{\sat{m}{A}}$ or $\satisfy[{\model[B]}]{v}{\sat{m}{B}}$, hence $\satisfy[{\model[B]}]{v}{\sat{m}{A} \OR \sat{m}{B}}$.
            This gives $\satisfy[{\model[B]}]{w}{\sat{m}{(A \OR B)} \IMP (\sat{m}{A} \OR \sat{m}{B})}$.

			\item $\jkax[4]: (\sat{m}{A} \IMP \just{t}{B}) \IMP \just{\jupdt{m}{t}}{(A \IMP B)}$.
			Let $v \irelop w$ with $\satisfy[{\model[B]}]{v}{\sat{m}{A} \IMP \just{t}{B}}$ i.e.~for all $u \irelop v$, $A \notin \val{\tm m}{u}$ or $B \in \val{\tm t}{u}$.
            By property 7 of Definition~\ref{def:basicmodmodel},
			$A \IMP B \in \val{(\jupdt{m}{t})}{v}$.
			So $\satisfy[{\model[B]}]{v}{\just{\jupdt{m}{t}}{(A \IMP B)}}$.
            Hence $\satisfy[{\model[B]}]{w}{(\sat{m}{A} \IMP \just{t}{B}) \IMP \just{\jupdt{m}{t}}{(A \IMP B)}}$.

			\item $\jkax[5]: \sat{m}{\BOT} \IMP \BOT$.
			Let $v \irelop w$, 
            by property 8 of Definition~\ref{def:basicmodmodel},
            we have $\BOT \notin \val{\tm m}{v}$ and so $\notsatisfy[{\model[B]}]{v}{\sat{m}{\BOT}}$.
			Therefore, $\satisfy[{\model[B]}]{w}{\sat{m}{\BOT} \IMP \BOT}$.

            \item $\just{\jbang[n]{c}}{\just{\jbang[n-1]{c}}{\just{\dots}{\just{\jbang{c}}{\just{c}{A}}}}}$ is a conclusion of $\can$.
            Then, as $\just{\jbang[n-1]{c}}{\just{\dots}{\just{\jbang{c}}{\just{c}{A}}}} \in \val{(\jbang[n]{c})}{w}$, we have that $\satisfy[{\model[B]}]{w}{\just{\jbang[n]{c}}{\just{\jbang[n-1]{c}}{\just{\dots}{\just{\jbang{c}}{\just{c}{A}}}}}}$.
            \qedhere
		\end{itemize}
	\end{proof}

    \begin{remark}
    To show soundness, it is important that conditions 1-4 in Def.~\ref{def:basicmodmodel}
    are inclusions and not equalities.
    As for example, were they equalities, the formula $(\sat{\sappl{t}{m}}{q} \AND \sat{m}{p}) \IMP \just{t}{(p \IMP q)}$ would be valid.
    However, it is not a theorem of $\JIKCS$ -- otherwise, by Theorem~\ref{thm:forgetful} we would get $\proves{\IK}{(\DIA q \AND \DIA p) \IMP \BOX(p \IMP q )}$ which is not true.
    We can build a countermodel $\model[B]$ to this formula meeting the inclusion-based conditions:
    $\worlds=\set{w}, \irel = \set{(w,w)}, \val{(\sappl{t}{m})}{w}= \set{q}, \val{m}{w} = \set{p}, \val{t}{w} = \emptyset$.
    Note that $\val{t}{w} \jappl{}{} \val{m}{w} = \emptyset \sset \val{(\sappl{t}{m})}{w}$.
    We have $\satisfy[{\model[B]}]{w}{\sat{\sappl{t}{m}}{q} \AND \sat{m}{p}}$ but $\notsatisfy[{\model[B]}]{w}{\just{t}{(p \IMP q)}}$, hence $\notsatisfy[{\model[B]}]{w}{ \sat{(\sappl{t}{m}}{q} \AND \sat{m}{p}) \IMP \just{t}{(p \IMP q)}}$.    
    \end{remark}

\subsection{Completeness}
In this section we prove the reverse direction of the last result, i.e.~completeness.

\begin{theorem}[Completeness]\label{thm:complbasicmodularmodel}
    For any formula $A$: $\basicmod {A} \implies \proves{\JIKCS}{A}$.
\end{theorem}

Our proof exploits the construction of a canonical model for intuitionistic justification logic.
This syntactic structure uses prime sets as worlds.

	\begin{definition}[Prime Set]
		A \emph{prime set} $\primeset \sset \langjust$ is a set of formulas satisfying the following:
		\begin{itemize}
			\item \emph{Deductive closure}: if $\proves[\primeset]{\JIKCS}{A}$ then $A \in \primeset$.
			
			\item \emph{Primeness}: if $A \OR B \in \primeset$ then $A \in \primeset$ or $B \in \primeset$.

			\item \emph{Consistency}: $\bot \notin \primeset$.
		\end{itemize}
	\end{definition}

    \begin{notation}
        Let $\primeset,\primeset[1] \sset \langjust$.
        We write $\proves[\primeset]{\JIKCS}{\primeset[1]}$
        if there is a finite set of formulas $\set{A_1, \dots, A_n}\subseteq\primeset[1]$ such that
        $\proves[\primeset]{\JIKCS}{A_1 \OR \dots \OR A_n}$.
        We also write $\notproves[\primeset]{\JIKCS}{\primeset[1]}$ when it does not hold that
        $\proves[\primeset]{\JIKCS}{\primeset[1]}$.
    \end{notation}

    The completeness proof relies on projecting unprovable consecutions
    into the canonical model.

\begin{lemma}[Prime Lemma]\label{lem:prime}
	Let $\primeset[0], \primeset[1] \sset \langjust$. 
    If $\notproves[{\primeset[0]}]{\JIKCS}{\primeset[1]}$
	then there exists a prime set $\primeset' \ssetop \primeset[0]$ such that $\notproves[\primeset']{\JIKCS}{\primeset[1]}$.
    Furthermore, $\primeset'$ is a maximal set (w.r.t.~inclusion) satisfying these conditions.
\end{lemma}
\begin{proof}[Proof sketch]
    The proof goes via a standard argument: we exploit an enumeration of 
    formulas to extend $\Gamma$ step-by-step
    while ensuring
    that the extension does not entail $\Delta$.
    The limit of this process can be shown to be a prime set $\Gamma'$ not 
    entailing $\Delta$.
    The maximality of $\Gamma'$ follows from the fact that
    if a formula could be added without entailing $\Delta$, 
    then it would indeed have been added when 
    enumerated.
\end{proof}

Next, we define a structure which we show to be a basic (modular) model.
To define parts of this structure, we use the sets of formulas
$\inv{t}{\primeset} := \set[A\in \langjust]{\just{\tm t}{A} \in \primeset}$
and 
$\inv{m}{\primeset} := \set[A\in \langjust] {\sat{\tm m}{A} \in \primeset}$.

\begin{definition}[Canonical basic model]\label{def:canbasmodmod}
		Let $\canmodel[B]$ be the tuple $(\canworlds, \canirel, \canbasicval)$  where:
        \begin{center}
        \begin{tabular}{r @{\hspace{0.2cm}} l @{\hspace{3cm}} r@{\hspace{0.2cm}}l}
         \raisebox{0.25ex}{\tiny$\bullet$} & $\canworlds \colonequals \set[\primeset]{\text{$\primeset$ is a prime set}}$; &   
         \raisebox{0.25ex}{\tiny$\bullet$} & $\canval{\tm t}{\primeset} \colonequals \inv{t}{\primeset}$; \\
         \raisebox{0.25ex}{\tiny$\bullet$} & $\canirel \colonequals \sset$; &   
         \raisebox{0.25ex}{\tiny$\bullet$} & $\canval{\tm m}{\primeset} \colonequals \inv{m}{\primeset}$. \\
         \raisebox{0.25ex}{\tiny$\bullet$} & $\canval{p}{\primeset} = 1$ iff $p \in \primeset$; &   
         &  \\
        \end{tabular}
        \end{center}

\end{definition}

\begin{proposition}\label{prop:canbasmodmod}
        $\canmodel[B]$ is a basic model.
	\end{proposition}

    \begin{proof}
        We are required to prove that this structure satisfies all 
        the properties from Definition~\ref{def:basicmodmodel}.
        As an example, we show that $\jappl{\canval{s}{\primeset}}{\canval{t}{\primeset}} \sset \canval{(\jappl{s}{t})}{\primeset}$.
        Let $B \in \jappl{\canval{s}{\primeset}}{\canval{t}{\primeset}}$.
			By definition, there is a formula $A$ such that $A \IMP B \in \canval{s}{\primeset}$ and $A \in \canval{t}{\primeset}$.
			This means that $\proves[\primeset]{\JIKCS}{\just{s}{(A \IMP B)}}$ and $\proves[\primeset]{\JIKCS}{\just{t}{A}}$.
			Using the $\jkax[1]$~axiom we get $\proves[\primeset]{\JIKCS}{\just{\jappl{s}{t}}{B}}$, hence $\just{\jappl{s}{t}}{B} \in \primeset$ by deductive closure.
			Therefore $B \in \inv{(\jappl{s}{t})}{\primeset} = \canval{(\jappl{s}{t})}{\primeset}$.
    \end{proof}

    \begin{lemma}[Truth Lemma]\label{lem:truth}
        Let $A$ be a formula and $\primeset$ be a prime set.
        Then:
		$A \in \primeset \iff \satisfy[{\canmodel[B]}]{\primeset}{A}$.
	\end{lemma}

    \begin{proof}
		By induction on $A$.
		All cases are covered in~\cite{marti_intutionistic_2016}, but the case for $A = \sat{\tm m}{B}$ which we show:
        \begin{center}
        \begin{tabular}{c c l}
          $\sat{\tm m}{B} \in \primeset$   & $\iff$ & $B \in \inv{m}{\primeset}$ by definition of $\inv{m}{\primeset}$ \\
             & $\iff$ & $B \in\canval{\tm m}{\primeset}$ by definition of the model $\canmodel[B]$ \\
             & $\iff$ & $\satisfy[{\canmodel[B]}]{\primeset}{\sat{\tm m}{B}}$ by definition of truth in a basic model.\\
        \end{tabular}
        \end{center}
        \vspace{-0.645cm}
        \qedhere
    \end{proof}

    With these results in hand, we can prove completeness.

    \begin{proof}[Proof of Theorem~\ref{thm:complbasicmodularmodel}]
		We prove the contrapositive. 
        Suppose $\notproves{\JIKCS}{A}$.
		Then by the Prime Lemma~\ref{lem:prime}, there exists a prime set $\primeset$ such that $\notproves[\primeset]{\JIKCS}{A}$.
		So $A \notin \primeset$, and by the Truth Lemma~\ref{lem:truth} we have $\notsatisfy[{\canmodel[B]}]{\primeset}{A}$.
        Finally, Proposition~\ref{prop:canbasmodmod} ensures that the $\canmodel[B]$ belongs to the adequate class of models.
	\end{proof}

\section{Modular Models}\label{sec:modular}
In the classical setting, modular models~\cite{artemov_ontology_2012} are epistemic models extending basic modular models with a modal accessibility relation to model the notion of \emph{knowledge}.
A key feature of these models is \emph{justification yields belief}, which says that if a formula of the form $\just{t}{A}$ holds at a world, then $A$ must hold in all accessible worlds.
This behaviour imitates that of the $\BOX$ operator in modal logic and hence the feature promotes ``backwards compatibility" from the justification logic to its corresponding modal logic.

In our setting, we must make several adaptations.
First, satisfaction of a modal formula $\BOX G$ is not defined locally, but globally along the intuitionistic pre-order, so the justification yields belief principle must account for this.
Second, to deal with satisfiers we introduce \emph{satisfaction yield possibility}, which says that if a formula of the form $\sat{m}{A}$ holds at a world, then $A$ must hold in some accessible world.
Finally, to promote ``backwards compatibility", we have to ensure that the additional frame conditions of a birelational frame (Definition~\ref{def:frames})~\cite{fischer_servi_axiomatizations_1984} also hold.

\subsection{Definition and Soundness}
We define models for intuitionistic justification logic 
that bridge it with intuitionistic modal logic via the modal relation.
To flesh out this connection, we make the interpretation of terms resonate with the modal relation by imposing conditions on the former using the latter and the notion of truth.
To get there, we need to introduce preliminary structures on which we evaluate formulas.

\begin{definition}[Intuitionistic quasi-model]
    An (intuitionistic) quasi-model is a tuple~$\model = (\worlds, \irel, \rel, \basicval)$ where $\model[B] = (\worlds, \irel, \basicval)$ is a basic model and $R$~is a binary relation on~$W$.
\end{definition}

\begin{definition}[Truth in a quasi-model]\label{def:quasi-model}
    For a quasi-model $\model = (\worlds, \irel, \rel, \basicval)$, a point $w\in W$ and a formula~$A$, we define the \emph{truth of $A$ at $w$ in $\model$} as:
    $\satisfy[{\model[M]}]{w}{A}$ iff
    $\satisfy[{\model[B]}]{w}{A}$
    where $\model[B] = (\worlds, \irel, \basicval)$.
\end{definition}

As the notion of truth is the same in basic and quasi-models, the following obviously follows.
    \begin{lemma}[Monotonicity Lemma]\label{lem:modmonotonicity}
		Let $\model = (\worlds, \irel, \rel,  \basicval)$ be a quasi-model. 
		Let $w,v\in\worlds$ such that $w \irel v$.
		Let $A$ be a formula.
		Then: $\satisfy{w}{A} \implies \satisfy{v}{A}$
	\end{lemma}
    \begin{proof}
        By definition $\satisfy{w}{A}$ implies $\satisfy[{\model[B]}]{w}{A}$ where $\model[B] = (\worlds, \irel,  \basicval)$.
        Using the Monotonicity Lemma \ref{lem:basmon} for basic modular models we get $\satisfy[{\model[B]}]{v}{A}$.
        By definition here again, we obtain $\satisfy{v}{A}$.
    \end{proof}

\begin{remark}
    Definition~\ref{def:quasi-model} exhibits the \emph{locality of truth}~\cite{marti_intutionistic_2016,kuznets_logics_2019}. 
    This notion in quasi-models helps us use these intermediate structures to 
    bridge justification and modal logics:
    we obtain the adequate models by restricting the interpretation of terms by $R$ and truth as shown in the next definition.
\end{remark}

\begin{definition}[Intuitionistic modular model]\label{def:modular-models}
        An (intuitionistic modular) model is a quasi-model $\model = (\worlds, \irel, \rel, \basicval)$ such that $\modmodel[M] \colonequals (\worlds, \irel, \rel, \basicval_\Prop)$ is a birelational intuitionistic model together with:

        \begin{itemize}
            \item \emph{Justification yields belief} (JYB) principle: 
            for all proof terms $\tm t \in \prfterms$
            \begin{center}
            $\val{\tm t}{w} \sset \jyb{w} \colonequals \set[A]{\text{for all $w', v' \in \worlds$ with $w \irel w' \rel v'$, $\satisfy{v'}{A}$}}$
            \end{center}
            
            \item \emph{Satisfaction yields possibility} (SYP) principle:
            for all satisfiers $\tm m \in \satterms$
            \begin{center}
            $\val{\tm m}{w} 
            \sset 
            \syp{w} 
            \colonequals 
            \set[A]{\text{there exists $v \in \worlds$ such that $w \rel v$ and $\satisfy{v}{A}$}}$
            \end{center}
        \end{itemize}

    We write $\modulmod A$ if for any model $\model[M]$ and world $w \in \worlds$ we have $\satisfy{w}{A}$.
\end{definition}

 In contrast to the classical setting or other intuitionistic variants~\cite{artemov_ontology_2012,marti_intutionistic_2016,pischke_intermediate_2023}, the definition of (JYB) takes into account the global interpretation of $\BOX$ via the pre-order~$\irel$.

Exploiting soundness with respect to basic models, the result follows for models.

\begin{theorem}[Soundness of modular models]
        For any formula $A$:
        $\proves{\JIKCS}{A} \implies \modulmod{A}$.
    \end{theorem}
    \begin{proof}
        Let $\model = (\worlds, \irel, \rel,  \basicval)$ be a modular model and
        $w$ a world. 
        Since $\model[B] = (\worlds, \irel,  \basicval)$ is a basic modular model, we have
        $\satisfy[{\model[B]}]{w}{A}$ by Theorem~\ref{thm:soundbasicmodularmodel}.
        Hence, $\satisfy{w}{A}$ by Definition~\ref{def:modular-models}. 
    \end{proof}

\begin{remark}[Differences between Modular models and Fitting models]
    The first epistemic models for classical justification logics were introduced by Fitting~\cite{fitting_logic_2005}.
    These differ from Artemov's modular models~\cite{artemov_ontology_2012} on the definition of satisfaction for justification formulas of the form $\just{t}{A}$.
    Fitting models additionally require $A$ to hold in all \emph{modally accessible} worlds, i.e.,
    $\satisfy{w}{\just{t}{A}}$ iff $A \in \val{\tm t}{w}$ and for all $v \in \worlds$ with $w \rel v$, $\satisfy{v}{A}$.
    In contrast, modular models make an \emph{ontological separation} with the modal accessibility relation which plays no role on the satisfaction of formulas. 
    Yet, both classes of models ensure ``backwards compatibility" with modal logic~\cite{kuznets_logics_2019};
    in fact, they are equivalent~\cite{kuznets_justifications_2012}.
\end{remark}
\subsection{Completeness}
We prove completeness w.r.t.~the semantics on models,
assuming that $\CS$ is axiomatically appropriate.

\begin{theorem}[Completeness]\label{thm:complmodularmodel}
    For any formula $A$: $\modulmod {A} \implies \proves{\JIKCS}{A}$.
\end{theorem}

We construct once more a canonical model
using the basic canonical model from Definition~\ref{def:canbasmodmod} as a basis.
In essence, we aim to show that an adequate modal relation
can be grafted onto this basic model.

We first need to define, for a given $\Gamma\subseteq\langjust$, certain sets of formulas:
$\hashprime[0] \colonequals \set[A]{\text{ for some } \tm t \in \prfterms,\,\just{t}{A} \in \primeset[0] }$
and $\flatprime[0] \colonequals \set[A]{\text{for all } \tm m \in \satterms,\,\sat{m}{A} \notin \primeset[0]}$.

\begin{definition}[Canonical intuitionistic modular model]\label{def:canmodmod}
		Let $\canmodel$ be the tuple $(\canworlds, \canirel, \canrel, \canbasicval)$ where
        $\canmodel[B] = (\canworlds, \canirel, \canbasicval)$ is the canonical basic model,
        and
		$\primeset[0] \canrel \primeset[1]$ 
        is defined as: 
        if $\just{\tm t}{A} \in \Gamma$ then $A\in \Delta$
        and for each $A \in \Delta$ there exists $\sat{\tm m}{A} \in \Gamma$.
	\end{definition}

It is immediate that $\canmodel$ is a quasi-model, as $\canmodel[B]$ is a basic model by Proposition~\ref{prop:canbasmodmod}.
Utilising the locality of truth, this allows us to directly prove the Truth Lemma for $\canmodel$.

\begin{lemma}[Truth Lemma]\label{lem:modtruth}
        Let $A$ be a formula and $\primeset$ be a prime set.
        Then: $A \in \primeset \iff
        \satisfy[{\canmodel[M]}]{\primeset}{A}$.
\end{lemma}

\begin{proof}
        As $\canmodel[B]$ is a basic model, we have
        $\satisfy[{\canmodel[M]}]{\primeset}{A}
            \iff
            \satisfy[{\canmodel[B]}]{\primeset}{A} \iff
            A \in \primeset$ by Lemma~\ref{lem:truth}.
    \end{proof}

The canonical model also displays the following strong properties which ensure all collections of formulas made true by the same satisfier $\tm m$ (or false by the same proof term $\tm t$) must be in fact true (or false, respectively) at the same world.

\begin{lemma}\label{lem:boxexist}\label{lem:extdiaexist}
        Let $\primeset[2] \sset \langjust$
        and
        $\primeset\in\canworlds$.
        \begin{enumerate}
            \item If 
            for all $\tm t \in \prfterms$ and for all $A_1, \dots, A_n \in \primeset[2]$,
            $\notproves[\primeset]{\JIKCS}{\just{t}{(A_1 \OR \dots \OR A_n)}}$,
            then there are $\primeset', \primeset[1]'\in\canworlds$ such that $\primeset \sset \primeset' \canrel \primeset[1]'$ and $\notproves[{\primeset[1]'}]{\JIKCS}{\primeset[2]}$.

            \item If for some $\tm m \in\satterms$ and for all $A_1, \dots, A_n \in \primeset[2]$,
            $\proves[\primeset]{\JIKCS}{\sat{m}{(A_1 \AND \dots \AND A_n)}}$,
            then there is $\primeset[1]\in\canworlds$ such that $\primeset \canrel \primeset[1]$ and $\proves[{\primeset[1]}]{\JIKCS}{A}$ for all $A\in\primeset[2]$. 
        \end{enumerate}        
    \end{lemma}
\begin{proof}
    1. Let $\tilde{\primeset[2]} = \set[\just{t}{(A_1 \OR \dots \OR A_n)}]{\text{$\tm{t}\in \prfterms$ and $A_1, \dots, A_n \in \primeset[2]$}}$.
    Since for all $\tm{t}\in\prfterms$ and $A_1, \dots, A_n \in \primeset[2]$ we have $\notproves[\primeset]{\JLCS}{\just{t}{(A_1 \OR \dots \OR A_n)}}$,
    we can establish that $\notproves[\primeset]{\JLCS}{\tilde{\primeset[2]}}$.
    Indeed, the primeness of $\Gamma$ informs us that if it entailed a 
    finite disjunction of elements of $\tilde{\primeset[2]}$
    then it would entail one of its disjuncts.
    Then, by the Prime Lemma~\ref{lem:prime} we get a maximal
    prime set $\primeset'$ with $\primeset \sset \primeset'$ and
    $\notproves[\primeset']{\JLCS}{\tilde{\primeset[2]}}$.
    We can then show that $\notproves[\primeset'^{\sharp}]{\JLCS}{\primeset[2]}$.
    Otherwise, there would be $C_1, \dots, C_m \in \primeset'^{\#}$ with each $\just{t_i}{C_i} \in \primeset'$ and $A_1, \dots, A_n \in \primeset[2]$ such that 
    $\proves[C_1, \dots, C_m]{\JLCS}{A_1 \OR \dots \OR A_n}$,
    which by the Lifting Lemma~\ref{lem:lifting}
    would give $\proves[\primeset']{\JLCS}{\just{t}{(A_1 \OR \dots \OR A_n)}}$
    for some $\tm t \in \prfterms$.
    Then, we can apply the Prime Lemma~\ref{lem:prime} again on $\notproves[\primeset'^{\sharp}]{\JLCS}{\primeset[2]}$ to construct a prime set $\primeset[1]'$ with $\primeset'^{\sharp} \sset \primeset[1]'$ and $\notproves[{\primeset[1]'}]{\JLCS}{\primeset[2]}$.
    
    Finally we need to check that $\primeset' \canrel \primeset[1]'$.
    Since $\primeset'^{\sharp} \sset \primeset[1]'$ by construction we have that if $\just{\tm t}{A} \in \Gamma'$ then $A\in \Delta'$, hence it is enough to show that 
    for each $A \in \Delta'$ there exists $\sat{\tm m}{A} \in \Gamma'$.
    For such $A \in \Delta'$, suppose for a contradiction that for all~$\tm m \in \satterms$, $\sat{m}{A} \notin \Gamma'$.
    Then, for all $A_1, \dots, A_n \in \primeset[2]$ we have $\notproves[\primeset', \sat{m}{A}]{\JLCS}{\just{t}{(A_1 \OR \dots \OR A_n)}}$
    for each $\tm{m}$ and $\tm{t}$,
    otherwise using the Deduction Theorem~\ref{thm:deduction} and $\jkax[4]$ we would have
    $\proves[\primeset']{\JLCS}{\just{\jupdt{m}{t}}{(A \IMP (A_1 \OR \dots \OR A_n))}}$
    and hence $A \IMP (A_1 \OR \dots \OR A_n) \in \Gamma'^{\#} \sset \Delta'$,
    so we could deduce from $A \in \Delta'$ and primeness of $\Delta'$ that 
    $A_i \in \Delta'$ for some $i\le n $.
    However, we have $\Gamma' \subsetneq \Gamma' \UNI \set{\sat{m}{A}}$ and so the Prime Lemma~\ref{lem:prime} would give a strictly larger prime set satisfying the same conditions as $\Gamma'$, contradicting its maximality
    -- a similar use of maximality to e.g.~\cite[Lemma~3 and Theorem~7]{fischer_servi_axiomatizations_1984}.
    
    2. 
    Let us fix $\tm m \in\satterms$ such that for all $A_1, \dots, A_n \in \primeset[2]$,
            $\proves[\primeset]{\JIKCS}{\sat{m}{(A_1 \AND \dots \AND A_n)}}$,
    hence by deductive closure of $\primeset$,
    $\sat{m}{(A_1 \AND \dots \AND A_n)}\in\primeset$.
    We first show that 
    $\notproves[{\primeset[2] \UNI \hashprime}]{\JLCS}{\flatprime}$.
    Otherwise, there would be $A_1, \dots, A_l \in \primeset[2]$, $B_1, \dots, B_n \in \hashprime$ with $\just{t_i}{B_i} \in \primeset$ for some $\tm{t_i}$, and $C_1, \dots, C_p \in \flatprime$ such that 
    $\proves[A_1, \dots, A_l, B_1, \dots, B_n]{\JLCS}{C_1 \OR \dots \OR C_p}$,
    and therefore $\proves[A_1 \AND \ldots \AND A_l, B_1, \dots, B_n]{\JLCS}{C_1 \OR \dots \OR C_p}$.
    By the Lifting Lemma~\ref{lem:lifting} this 
    yields some $\tm{n}\in\satterms$ such that $\proves[\sat{m}{(A_1 \AND \ldots \AND A_l)}, \just{t_1}{B_1}, \dots, \just{t_n}{B_n}]{\JLCS}{\sat{n}{(C_1 \OR \dots \OR C_p)}}$, hence
    $\proves[\primeset]{\JLCS}{\sat{n}{(C_1 \OR \dots \OR C_p)}}$.
    Now, using the $\jkax[3]$ axiom as well as the deductive closure and primeness of $\primeset$, we get
    $\sat{\tm n}{C_j}\in\Gamma$ for some $j\le p$, which contradicts $C_j\in\flatprime$.
    Then, we can apply the Prime Lemma~\ref{lem:prime} on $\notproves[{\primeset[2] \UNI \hashprime}]{\JLCS}{\flatprime}$
    to get a prime set $\primeset[1] \ssetop \primeset[2] \UNI \hashprime$ with
    $\notproves[{\primeset[1]}]{\JLCS}{\flatprime}$.
    Finally, $\primeset \canrel \primeset[1]$ because as $\primeset^{\sharp} \sset \primeset[1]$ by construction we have that if $\just{\tm t}{A} \in \Gamma$ then $A\in \Delta$, and as $\notproves[{\primeset[1]}]{\JLCS}{\flatprime}$ we have that
    for each $A \in \Delta$ there must exist $\sat{\tm m}{A} \in \Gamma$.
\end{proof}

As our next step in showing that $\canmodel$ is an intuitionistic modular model,
we show that $\canmodmodel$ is an intuitionistic birelational model.
We therefore proceed to establish the confluence properties.

    \begin{lemma}[Forwards and backwards confluence in $\canmodel$]\label{lem:canforconf}\label{lem:canbacconf}
        Let $\primeset, \primeset', \primeset[1]\in\canworlds$.
        \begin{enumerate}
            \item[(FC)]
            If $\primeset \sset \primeset'$ and $\primeset \canrel \primeset[1]$,
            then there exists $\primeset[1]'\in\canworlds$ with $\primeset[1] \sset \primeset[1]'$ and $\primeset' \canrel \primeset[1]'$.

            \item[(BC)]
            If $\primeset \canrel \primeset[1] \sset \primeset[1]'$,
            then there exists $\primeset'\in\canworlds$ with $\primeset \sset \primeset' \canrel \primeset[1]'$.
        \end{enumerate}
        
    \end{lemma}

    \begin{proof}
    While our proof is similar to the one for modal birelational models for $\IK$ in~\cite{fischer_servi_axiomatizations_1984},
    there are some subtleties in handling the satisfier terms.
    We expand on this point in the remark below. 
    \end{proof}

    \begin{remark}\label{rem:non-normal}
        The use of proof terms in classical justification logic is known to introduce some non-normal behaviours of the modality~\cite{artemov_why_2011,artemov_justification_2024}.

        Similarly, 
        in the intuitionistic setting,
        the $\IK$ theorem $\DIA(G_1 \AND \dots \AND G_n) \IMP (\DIA G_1 \AND \dots \AND \DIA G_n)$ does not distinguish the information the distinct diamonds are carrying. 
        The justification version of this theorem, 
        $\sat{m}{(A_1 \AND \dots \AND A_n)} \IMP (\sat{\sappl{t_1}{m}}{A_1} \AND \dots \AND \sat{\sappl{t_n}{m}}{A_n})$
        where $\tm m$ is a satisfier variable and each $\tm {t_i}$ is a ground term obtained by applying the Lifting Lemma~\ref{lem:lifting} to $\proves{\IK}{(A_1 \AND \dots \AND A_n) \IMP A_i}$,
        is more constrained due to the specific structure of the $\sappl{}{}$-satisfier terms on the right-hand-side.
    \end{remark}

We finally establish that $\canmodel$ is a modular model 
by proving that it obeys the JYB and SYP principles.
\emph{Justification yields belief} is already achieved from the definition of $\canrel$,
following the lines of~\cite{marti_intutionistic_2016}.
    \begin{lemma}[JYB in $\canmodel$]\label{lem:canJYB}
        Let $\tm t\in\prfterms$ and $\primeset\in \canworlds$.
        Then: 
        
        $A\in\canval{\tm t}{\primeset} \implies$ for all $\primeset',\primeset[1]'\in\canworlds$ such that $\primeset \sset \primeset' \canrel \primeset[1]'$ we have $\satisfy[{\canmodel[M]}]{\primeset[1]'}{A}$.
    \end{lemma}

    \begin{proof}
        Suppose $A \in \canval{\tm t}{\primeset}$, then $\just{t}{A} \in \primeset$ by Definition~\ref{def:canbasmodmod}. 
        So, for $\primeset' \ssetop \primeset$ we have $\just{t}{A} \in \primeset'$ too.
        Hence, for $\primeset[1]'$ with $\primeset' \canrel \primeset[1]'$, we have $A \in \primeset[1]'$ by Definition~\ref{def:canmodmod}.
        By the Truth Lemma~\ref{lem:modtruth}, we have $\satisfy[{\canmodel[M]}]{\primeset[1]'}{A}$.
    \end{proof}

    The remaining principle of \emph{satisfaction yields possibility} is slightly more involved, as shown below.

    \begin{lemma}[SYP in $\canmodel$]\label{lem:canSYP}
        Let $\tm m\in\satterms$ and $\primeset\in\canworlds$. Then:
        
        $A \in \canval{\tm m}{\primeset} \implies$ there exists $\primeset[1] \in \canworlds$ such that $\primeset \canrel \primeset[1]$ and $\satisfy[{\canmodel[M]}]{\primeset[1]}{A}$.
    \end{lemma}

    \begin{proof}
    Suppose $A \in \canval{\tm m}{\primeset}$, then $\sat{\tm m}{A} \in \primeset$ by definition.
    So, $\proves[\primeset]{\JIKCS}{\sat{\tm m}{A}}$.
    Therefore, by Lemma~\ref{lem:extdiaexist} there is $\primeset[1]\in\canworlds$ with $\primeset \canrel \primeset[1]$ and $\proves[{\primeset[1]}]{\JIKCS}{A}$.
    By deductive closure $A \in \primeset[1]$ and apply the Truth Lemma~\ref{lem:modtruth} for $\canmodel$ to get $\satisfy[{\canmodel[M]}]{\primeset[1]}{A}$.
    \end{proof}

The next proposition reveals the nature of $\canmodel$ by gathering our results.

\begin{proposition}\label{prop:canmodmod}
    $\canmodel$ is an intuitionistic modular model.    
\end{proposition}

\begin{proof}
    First, $\canmodel$ is a quasi-model as $\canmodel[B]$ is a basic model by Proposition~\ref{prop:canbasmodmod}.
    Then, Lemma~\ref{lem:canforconf}
    shows that forwards and backwards confluence hold in $\canmodel$.
    Finally, Lemma~\ref{lem:canJYB} and Lemma~\ref{lem:canSYP}
    respectively prove that the properties JYB and SYP are satisfied in $\canmodel$.
\end{proof}

With this in hand, completeness follows.

    \begin{proof}[Proof of Theorem~\ref{thm:complmodularmodel}]
        We reason contrapositively and assume $\notproves{\JIKCS}{A}$.
		The Prime Lemma~\ref{lem:prime} informs us of the existence of a prime set $\primeset$ such that $\notproves[\primeset]{\JIKCS}{A}$.
		So $A \notin \primeset$, which gives $\notsatisfy[{\canmodel[M]}]{\primeset}{A}$
        by the Truth Lemma~\ref{lem:modtruth}.
        Finally, Proposition~\ref{prop:canmodmod} ensures that the $\canmodel$ belongs to the adequate class of models.
    \end{proof}

\section{Realisation}\label{sec:real}
We recall that our goal is to prove the following.
\begin{reptheorem}{thm:realisation}
    For schematic $\CS$, there is a realisation map $r$ such that for all $G\in\langmod$:
         $$\proves{\IK}{G} \implies \proves{\JIKCS}{\real{G}}.$$
\end{reptheorem}

This realisation function is constructed methodically in two steps:
\begin{enumerate}
\item Firstly, we reason semantically that a theorem $G$ of $\IK$ can be \emph{pre-realised} into a disjunction of justification formulas which have ``roughly'' the same structure as $G$;
\item Secondly, we \emph{condense} a pre-realisation into a realisation by syntactical reasoning and substitutions to remove the extraneous content.
\end{enumerate}

Our goal is more specifically to define \emph{normal} realisation map which differentiates modalities based on their polarity.
Recall that
a subformula of $G$ is \emph{positive} if its position in the formula tree of $G$ is reached from the root by following the left branch of an implication an even number of times; and call it \emph{negative} otherwise.

\begin{definition}[Normal realisation]
A realisation map 
$\function{\realfunction}{\langmod}{\langjust}$
is \emph{normal} if,
for any $G \in \langmod$,
its negative modal subformulas $\BOX F$ (or $\DIA F$) are realised as
$\just{\tm x}{B}$ (or~$\sat{\tm a}{B}$) with $\tm x\in\prfvars$ (or~$\tm a\in\satvars$, respectively)
occurring in $\real G$ exactly once.
\end{definition}

\begin{remark}
The choice of a normal realisation 
ensures that no independent subformulas are unintensionally identified.
For example, the $\IK$-theorem $(\DIA p \OR \DIA q) \IMP \DIA (p \OR q)$ can be realised as
    $$
        (\sat{a}{p} \OR \sat{a}{q}) \IMP \sat{\suni{\sappl{t_1}{a}}{\sappl{t_2}{a}}}{(p \OR q)}
        \quad\text{or}\quad
        (\sat{a}{p} \OR \sat{b}{q}) \IMP \sat{\suni{\sappl{t_1}{a}}{\sappl{t_2}{b}}}{(p \OR q)}
    $$
    where $\tm{t_1}$ and $\tm{t_2}$ are ground terms
    obtained by applying the Lifting Lemma~\ref{lem:lifting} to $\proves{\JIKCS}{p \IMP (p \OR q)}$ and $\proves{\JIKCS}{q \IMP (p \OR q)}$, respectively.
    But, the first instance cannot be the result of a normal realisation. 
\end{remark}

For this reason, we need to do some bookkeeping on $\BOX$s and $\DIA$s:
depending on their (positive or negative) polarity, they will be replaced by a variable in $\prfvars \UNI \satvars$ or a compound term in $\prftm \UNI \satterms$.
So, we first annotate by natural numbers
$\BOX$s and $\DIA$s in the modal logic formulas.
Each annotation is unique. Formally:

\begin{definition}[Annotated modal formulas]
    An \emph{annotated modal formula} is built up like a standard modal formula but instead of single modalities $\BOX$ and $\DIA$, there is an infinite family of indexed modalities $\BOX[k]$ and $\DIA[k]$, for $k\in\mathbb N$.
    We assume that
    no index occurs more than once in any given annotated formula.
\end{definition}

\begin{example}
    Given the modal formula $\BOX\DIA G \IMP \BOX(\DIA J \AND \DIA \BOX (H \OR I))$, a possible annotated version is 
    $\BOX[1]\DIA[2] G \IMP \BOX[3](\DIA[4] J \AND \DIA[5] \BOX[6] (H \OR I))$.
\end{example}

\noindent
Additionally, we need to track the signs of modal formulas.
\begin{definition}[Signed modal formulas]
    \emph{Signed modal formulas} are annotated modal formulas assigned with a polarity $\tsign$ or $\fsign$.
    We denote this as~$\T{G}$ or~$\F{G}$.
    The set of signed formulas is $\langmodann$.
\end{definition}

We stress that these annotations do not add semantic information: they are bookkeeping devices. 
Where it is clear, we do not distinguish between a modal formula and its annotated version.
\subsection{Pre-realisation}
In this subsection, 
we adapt the methodology of~\cite{fitting_logic_2005,kuznets_logics_2019,pischke_intermediate_2023} 
and reason semantically on both the canonical model $\canmodel$ and its induced birelational intuitionistic model~$\canmodmodel$ (recall Definition~\ref{def:modular-models}) to \emph{pre-realise} modal formulas.

\begin{definition}[Pre-realiser]\label{def:quasi-realiser}
        The \emph{pre-realiser} is a function
        $\function{\qrel{\cdot}}{\langmodann}{\pset{\langjust}}$
        defined as:
        \scalebox{0.95}{
        $\begin{array}{c @{\hspace{0.02cm}}c @{\hspace{0.02cm}} l@{\quad}c @{\hspace{0.02cm}} c @{\hspace{0.02cm}} l}
            \qrel{\T{p}} & \colonequals & \set{p} &
            \qrel{\F{p}} & \colonequals & \set{p}
            \\
            \qrel{\T{G \AND H}} & \colonequals & \set[A \AND B]{A \in \qrel{\T{G}}, B \in \qrel{\T{H}}} &
            \qrel{\F{G \AND H}} & \colonequals & \set[A \AND B]{A \in \qrel{\F{G}}, B \in \qrel{\F{H}}}
            \\
            \qrel{\T{G \OR H}} & \colonequals & \set[A \OR B]{A \in \qrel{\T{G}}, B \in \qrel{\T{H}}} &
            \qrel{\F{G \OR H}} & \colonequals & \set[A \OR B]{A \in \qrel{\F{G}}, B \in \qrel{\F{H}}}
             \\
             \qrel{\T{G \IMP H}} & \colonequals & \set[A \IMP B]{A \in \qrel{\F{G}}, B \in \qrel{\T{H}}} &
             \qrel{\F{G \IMP H}} & \colonequals & 
             \multirow{2}{6cm}{$\set[(A_1 \AND \dots \AND A_n) \IMP (B_1 \OR \dots \OR B_m)]{\phantom{A}\,\,A_i \in \qrel{\T{G}}, B_j \in \qrel{\F{H}}}$}
             \\
              &&& \\
             \qrel{\T{\BOX[n] G}} & \colonequals & \set[\just{\pvar[n]}{A}]{A \in \qrel{\T{G}}} &
             \qrel{\F{\BOX[n] G}} & \colonequals & \multirow{2}{5cm}{$\set[\just{t}{(A_1 \OR \dots \OR A_m)}]{A_i \in \qrel{\F{G}}, 
             \newline\phantom{A}\,\tm t \in \prfterms}$}
             \\
              &&& \\
             \qrel{\T{\DIA[n] G}} & \colonequals & \set[\sat{\svar[n]}{(A_1 \AND \dots \AND A_m)}]{A_i \in \qrel{\T{G}}} &
             \qrel{\F{\DIA[n] G}} & \colonequals & \set[\sat{\tm m}{A}]{A \in \qrel{\F{G}}, \tm m \in \satterms}
        \end{array}$}
    \end{definition}
    
    \begin{definition}
        A modal formula~$G$ is \emph{pre-realisable} in $\JIKCS$ if 
        $\proves{\JIKCS}{\qrel{\F{G}}}$, 
        i.e.,
        there exist some $A_1, \dots, A_n \in \qrel{\F{G}}$ such that $\proves{\JIKCS}{A_1 \OR \dots \OR A_n}$.
    \end{definition}

     The following is the key result of this subsection 
    and corresponds to the first step of the realisation proof described above and provides an initial but imperfect connection between theorems of $\IK$ and theorems of $\JIKCS$. 
    It leverages the semantic machinery developed up to now.

    \begin{theorem}[Pre-realisation]\label{thm:quasreal}
        Let $\CS$ be axiomatically appropriate. For $G\in\langmod$,
        if  $\proves{\IK}{G}$,
        then $G$ is \emph{pre-realisable} in $\JIKCS$.
    \end{theorem}

This is shown through
the semantic connection between a pre-realiser and its modal formula which is understood in the following:
    \begin{theorem}\label{thm:qrelvalid}
        For every annotated modal formula~$G$.
        For a prime set $\primeset$.
        \begin{enumerate}
            \item If for all $A \in \qrel{\T{G}}$ we have
            $\satisfy[\canmodel]{\primeset}{A}$, then $\satisfy[\canmodmodel]{\primeset}{G}$.
            \item If for all $A \in \qrel{\F{G}}$ we have
            $\notsatisfy[\canmodel]{\primeset}{A}$, then $\notsatisfy[\canmodmodel]{\primeset}{G}$.
        \end{enumerate}
    \end{theorem}

    \begin{proof}
        We prove both simultaneously by induction on~$G$.
        The cases when $G = p, H \AND I, H \OR I, H \IMP I$ follow the argument in~\cite{pischke_intermediate_2023}.
        We extend the proof with an adapted case for $\BOX$ and a new one for $\DIA$.

        If $G = \BOX[n] H$ then $A\in\qrel{\T{G}}$ is of the form $\just{x_n}{B}$ for  $B\in\qrel{\T{H}}$ and $A\in\qrel{\F{G}}$ is of the form $\just{t}{(B_1 \OR \dots \OR B_m)}$ for $B_1,\dots,B_m\in\qrel{\F{H}}$ and $\tm t\in\prfterms$.

        \begin{enumerate}
        	\item 
        	Suppose for all $A \in \qrel{\T{G}}$, 
            $\satisfy[\canmodel]{\primeset}{A}$, that is, 
            for all $B \in \qrel{\T{H}}$, $\satisfy[\canmodel]{\primeset}{\just{\pvar[n]}{B}}$.
        	So, $B\in\canval{\pvar[n]}{\primeset}$ and by Lemma~\ref{lem:canJYB}, for all $\primeset', \primeset[1]'$ with $\primeset \sset \primeset' \canrel \primeset[1]'$, $\satisfy[\canmodel]{\primeset[1]'}{B}$.
        	By (IH) on 1., $\satisfy[\canmodmodel]{\primeset[1]'}{H}$ and therefore $\satisfy[\canmodmodel]{\primeset}{\BOX[n] H}$.
        	
        	\item 
            Suppose for all $A \in \qrel{\F{G}}$, 
            $\notsatisfy[\canmodel]{\primeset}{A}$, that is, 
            for all $B_1, \dots, B_m \in \qrel{\F{H}}$ and $\tm t \in \prfterms$, $\notsatisfy[\canmodel]{\primeset}{\just{t}{(B_1 \OR \dots \OR B_m)}}$.
            So by the Truth Lemma~\ref{lem:modtruth}, $\just{t}{(B_1 \OR \dots \OR B_m)}\notin\primeset$, hence $\notproves[\primeset]{\JLCS}{\just{t}{(B_1 \OR \dots \OR B_m)}}$.
            By Lemma~\ref{lem:boxexist}, there exist $\primeset', \primeset[1]'\in\canworlds$ such that $\primeset \sset \primeset' \canrel \primeset[1]'$ and $\notproves[{\primeset[1]'}]{\JLCS}{\qrel{\F{H}}}$.
            So for all $B \in \qrel{\F{H}}$, $B\notin\primeset[1]'$ and by the Truth Lemma~\ref{lem:modtruth} again, $\notsatisfy[\canmodel]{\primeset[1]'}{B}$ .
            By (IH) on 2., $\notsatisfy[\canmodmodel]{\primeset[1]'}{H}$, hence $\notsatisfy[\canmodmodel]{\primeset}{\BOX[n] H}$.
        \end{enumerate}

        If $G = \DIA[n] H$ then $A\in\qrel{\T{G}}$ is of the form $\sat{\svar[n]}{(B_1 \AND \dots \AND B_m)}$ for $B_1,\dots,B_m\in\qrel{\T{H}}$, and $A\in\qrel{\F{G}}$ is of the form $\sat{m}{B}$ for  $B\in\qrel{\F{H}}$ and $\tm m\in\satterms$.
    
        \begin{enumerate}
            \item Suppose for all $A \in \qrel{\T{G}}$, 
            $\satisfy[\canmodel]{\primeset}{A}$, that is, 
            for all $B_1, \dots, B_m \in \qrel{\T{H}}$, $\satisfy[\canmodel]{\primeset}{\sat{\svar[n]}{(B_1 \AND \dots \AND B_n)}}$.
            By Lemma~\ref{lem:extdiaexist}, there exists $\primeset[1]$ with $\primeset \canrel \primeset[1]$ and $\proves[{\primeset[1]}]{\JLCS}{B}$ for all $B \in \qrel{\T{H}}$.
            By deductive closure $B \in \primeset[1]$ and by the Truth Lemma~\ref{lem:modtruth} we get $\satisfy[{\canmodel[M]}]{\primeset[1]}{B}$ for all $B \in \qrel{\T{H}}$.
            Hence, by (IH) on 1., $\satisfy[\canmodmodel]{\primeset[1]}{H}$ and so $\satisfy[\canmodmodel]{\primeset}{\DIA[n] H}$.

            \item Suppose for all $A \in \qrel{\F{G}}$, 
            $\notsatisfy[\canmodel]{\primeset}{A}$, that is, 
            for all $B \in \qrel{\F{H}}$ and $\tm m \in \satterms$, $\notsatisfy[\canmodel]{\primeset}{\sat{m}{B}}$,
            and by the Truth Lemma~\ref{lem:modtruth}, $\sat{\tm m}{B} \notin \primeset$.
            By definition of $\canrel$ (Definition~\ref{def:canmodmod}), 
            it means, for all $\primeset[1]\in\canworlds$ with $\primeset \canrel \primeset[1]$, that $B \notin \primeset[1]$, 
            and by the Truth Lemma~\ref{lem:modtruth} again, $\notsatisfy[\canmodel]{\primeset[1]}{B}$ for all $B \in \qrel{\F{H}}$.
            Hence, by (IH) on 2., $\notsatisfy[\canmodmodel]{\primeset[1]}{H}$ and so $\notsatisfy[\canmodmodel]{\primeset}{\DIA[n] H}$.
            \qedhere
        \end{enumerate}
    \end{proof}

Now we can prove the Pre-realisation Theorem~\ref{thm:quasreal}.
    \begin{proof}[Proof of Theorem~\ref{thm:quasreal}]
        Suppose $G$ is not pre-realisable.
        Then $\notproves{\JIKCS}{\qrel{\F{G}}}$.
        By the Prime Lemma~\ref{lem:prime}, there exists a prime set $\primeset$ such that $\notproves[\primeset]{\JIKCS}{\qrel{\F{G}}}$.
        Hence, for all $A \in \qrel{\F{G}}$, $A \notin\primeset$, so
        by the Truth Lemma~\ref{lem:modtruth}, $\notsatisfy[\canmodel]{\primeset}{A}$.
        By Theorem~\ref{thm:qrelvalid}, $\notsatisfy[\canmodmodel]{\primeset}{G}$, so $G$ is not a theorem of~$\IK$.
    \end{proof}
\subsection{From Pre-realisation to Realisation}
This subsection performs the second step of the realisation proof sketched above. 
From any finite subset of pre-realisers for a given modal formula $G$, through a process called \emph{condensing}~\cite{fitting_modal_2016} which syntactically 
removes the extra disjunctions and conjunctions, 
we construct a justification formula that \emph{potentially} realises $G$, i.e.~respects its syntax but not necessarily its semantics.
Eventually, the realisation map, will choose a valid justification formula from the set of potential realisers, guided by the Pre-realisation theorem~\ref{thm:quasreal} from the previous section.

    \begin{definition}[Potential realiser]
        The \emph{potential realiser} is a map
        $\function{\prel{\cdot}}{\langmodann}{\pset{\langjust}}$
        defined as:

        $\begin{array}{c@{\colonequals}l@{\quad}c@{\colonequals}l}
            \prel{\T{p}} &  \set{p} &
            \prel{\F{p}} & \set{p}
            \\
            \prel{\T{G \AND H}} & \set[A \AND B]{A \in \prel{\T{G}}, B \in \prel{\T{H}}} &
            \prel{\F{G \AND H}} & \set[A \AND B]{A \in \prel{\F{G}}, B \in \prel{\F{H}}}
            \\
            \prel{\T{G \OR H}} & \set[A \OR B]{A \in \prel{\T{G}}, B \in \prel{\T{H}}} &
            \prel{\F{G \OR H}} & \set[A \OR B]{A \in \prel{\F{G}}, B \in \prel{\F{H}}}
             \\
             \prel{\T{G \IMP H}} & \set[A \IMP B]{A \in \prel{\F{G}}, B \in \prel{\T{H}}} &
             \prel{\F{G \IMP H}} & \set[A \IMP B]{A \in \prel{\F{G}}, B \in \prel{\T{H}}}
             \\
             \prel{\T{\BOX[n] G}} & \set[\just{\pvar[n]}{A}]{A \in \prel{\T{G}}} &
             \prel{\F{\BOX[n] G}} & \set[\just{t}{A}]{A \in \prel{\F{G}}, \tm t \in \prfterms}
             \\
             \prel{\T{\DIA[n] G}} & \set[\sat{\svar[n]}{A}]{A \in \prel{\T{G}}} &
             \prel{\F{\DIA[n] G}} & \set[\sat{\tm m}{A}]{A \in \prel{\F{G}}, \tm m \in \satterms}
        \end{array}$
    \end{definition}

The idea is that a realisation function on a modal formula $G$ would choose a suitable justification formula $A \in \prel{\F{G}}$.
The choice of a formula from $\prel{\F{G}}$, rather than $\prel{\T{G}}$, ensures that for any negative subformula of $A \in \prel{\F{G}}$ of the shape $\just{t}{B}$ (or~$\sat{m}{B}$), we have $\tm t \in \prfvars$ (or $\tm m \in \satvars$ respectively), which is required to reach a normal realisation.
The difference in the definition of $\prel{\F{G \IMP H}}$, $\prel{F{\BOX[n] G}}$ and $\prel{\T{\DIA[n] G}}$ compared to the pre-realiser ensures the justification formula does not contain additional conjunctions and disjunctions and exactly matches the structure of the modal formula.
More formally:
\begin{proposition}\label{prop:prelforget}
    Let $G$ be an annotated modal formula.
    Then for every $A \in \prel{\T{G}} \UNI \prel{\F{G}}$, $\forget{A} = G$,
    where $\function{\forget{(\cdot)}}{\langjust}{\langmod}$ is the forgetful projection from Definition~\ref{def:forgetful}.
\end{proposition}
\begin{proof}
    This is a routine proof by induction on $G$.
\end{proof}

The pre-realiser from Definition~\ref{def:quasi-realiser} is converted into a potential realiser using substitutions and syntactic reasoning -- this process is called \emph{condensing} by Fitting~\cite[p.~640]{fitting_modal_2016}.

    To construct adequate substitutions we use the following auxiliary definitions.

\begin{definition}[No new variable condition]
        Let $\sub$ be a substitution.
		The \emph{domain} of~$\sub$ is the set
        $\dom{\sub} \colonequals 
        \set[\pvar \in \prfvars]{\subst{\pvar} \neq \pvar}
        \UNI
        \set[\svar \in \satvars]{\subst{\svar} \neq \svar}$.

        A substitution $\sub$ \emph{meets the no new variable condition} if for each variable $\pvar\in \dom{\sub}$ (or $\svar\in \dom{\sub}$), $\subst{\pvar}$ (or $\subst{\svar}$) contains no other variables than $\pvar$ (or $\svar$).
\end{definition}

\begin{definition}[Lives on]
        Let $G$ be an annotated modal formula
        and $\sub$ be a substitution.
        We say that 
        $\sub$ \emph{lives on} $G$ if for every $\pvar[k] \in \dom{\sub}$, $\BOX[k]$ occurs in $G$ and for every $\svar[k] \in \dom{\sub}$, $\DIA[k]$ occurs in $G$.
    \end{definition}

With these definitions in hand, we tackle the condensing theorem.

\begin{theorem}[Condensing Theorem]\label{thm:condensing}
        Let $G\in\langmodann$.
        \begin{enumerate}
            \item For each collection $A_1, \dots, A_n \in \qrel{\T{G}}$, there exist $A \in \prel{\T{G}}$ and $\sub$ that lives on $G$ and meets the no new variable condition such that $A \proves{\JIKCS}{\subst{(A_1 \AND \dots \AND A_n)}}$.

            \item For each collection $A_1, \dots, A_n \in \qrel{\F{G}}$, there exist $A \in \prel{\F{G}}$ and $\sub$ that lives on $G$ and meets the no new variable condition such that $\subst{(A_1 \OR \dots \OR A_n)} \proves{\JIKCS}{A}$.
        \end{enumerate}
    \end{theorem}

\begin{proof}
    We prove both simultaneously by induction on $G$.
    Except for diamonds, all cases are covered by Pischke~\cite{pischke_intermediate_2023}.
    So, we are left to consider the case when $G = \DIA[k] H$.

    1.~Let $\sat{\svar[k]}{(A_1^1 \AND \dots \AND A_{n_1}^1)}, \dots, \sat{\svar[k]}{(A_1^m \AND \dots \AND A_{n_m}^m)} \in \qrel{\T{\DIA[k]H}}$ with $A_1^1, \dots, A_{n_1}^1, \dots, A_1^m, \dots, A_{n_m}^m \in \qrel{\T{H}}$.
    By the inductive hypothesis, there exist $A \in \prel{\T{H}}$ and $\sub$ living on $H$ which meets the no new variable condition such that
    $A\proves{\JIKCS}{\subst{(A_1^1 \AND \dots \AND A_{n_1}^1 \AND \dots \AND A_1^m \AND \dots \AND A_{n_m}^m )}}$.
    So, for each $i \le m$,
    $\proves{\JIKCS}{A \IMP \subst{(A_1^i \AND \dots \AND A_{n_i}^i)}}$.
    By the Lifting Lemma~\ref{lem:lifting}, there exist ground terms $\tm t_1, \dots, \tm t_n \in \prfterms$ such that
    $\proves{\JIKCS}{\just{t_i}{(A \IMP \subst{(A_1^i \AND \dots \AND A_{n_i}^i)})}}$.
    Set $\tm t \colonequals \jsum{t_1}{\jsum{\dots}{t_m}}$ which is also a ground term, and by repetitive use of $\jsumaxl$ and $\jsumaxr$, we have
    $\proves{\JIKCS}{\just{t}{(A \IMP \subst{(A_1^i \AND \dots \AND A_{n_i}^i)})}}$
    for all $i \le m$.

    Now, define $\sub'$
    as $\sub'(\svar[k]) = \sappl{t}{\svar[k]}$ and as the identity otherwise -- note as $\tm t$ is a ground term, $\sub'$~meets the no new variable condition.
    By the Substitution Lemma~\ref{lem:substitution}, Definition~\ref{def:substitution} and the fact that $\tm t$ is a ground term, so $\subst[\sub']{\tm t} = \tm t$, we have
    $\proves{\JIKCS}{\just{t}{(\subst[\sub']{A} \IMP \subst[\sub']{\subst{(A_1^i \AND \dots \AND A_{n_i}^i)}})}}$
    for all $i = 1, \dots, m$.
    Using the $\jkax[2]$ axiom and $\modus$:
    $\proves{\JIKCS}{\sat{\svar[k]}{\subst[\sub']{A}} \IMP \sat{\sappl{t}{\svar[k]}}{\subst[\sub']{\subst{(A_1^i \AND \dots \AND A_{n_i}^i)}}}}$.
    Which can be rewritten as
    $\sat{\svar[k]}{(\subst[\sub']{A})}\proves{\JIKCS}{\subst[\sub']{\subst{(\sat{\svar[k]}{(A_1^i \AND \dots \AND A_{n_i}^i)})}}}$
    since $\subst[\sub']{\subst{\svar[k]}}= \subst[\sub']{\svar[k]} = \sappl{t}{\svar[k]}$.
    Propositional reasoning gives
    $\sat{\svar[k]}{(\subst[\sub']{A})}\proves{\JIKCS}{\subst[\sub']{\subst{(\sat{\svar[k]}{(A_1^1 \AND \dots \AND A_{n_1}^1)} \AND \dots \AND \sat{\svar[k]}{(A_1^m \AND \dots \AND A_{n_m}^m)})}}}$
    where indeed $\sub \sub'$ is a substitution which lives on~$\DIA[k] H$, meets the no new variable condition and $\sat{\svar[k]}{(\subst[\sub']{A})} \in \prel{\T{\DIA[k] H}}$.

    2.~Let $\sat{\tm m_1}{A_1}, \dots, \sat{\tm m_n}{A_n} \in \qrel{\F{\DIA[k] H}}$ with $A_1, \dots, A_n \in \qrel{\F{H}}$ and $\tm m_1, \dots, \tm m_n \in \satterms$.
    By Proposition~\ref{prop:orterms}, there exists a satisfier $\tm m \in \satterms$ such that
    $(*)$ $\sat{\tm m_1}{A_1} \OR \dots \OR \sat{\tm m_n}{A_n} \proves{\JIKCS}{\sat{\tm m}{(A_1 \OR \dots \OR A_n)}}$.
    By the inductive hypothesis, there exist $A \in \prel{\F{H}}$ and $\sub$ that lives on $H$ (and meets the no new variable condition) such that $(**)$
    $\subst{(A_1 \OR \dots \OR A_n)} \proves{\JIKCS}{A}$.
    Note that by definition $\sub$ also lives on~$\DIA[k] H$ trivially.
    By the Substitution Lemma~\ref{lem:substitution} applied to $(*)$, we have
    $\subst{(\sat{\tm m_1}{A_1} \OR \dots \OR \sat{\tm m_n}{A_n})}\proves{\JIKCS}{\subst{(\sat{\tm m}{(A_1 \OR \dots \OR A_n)})}}$.
    By the Lifting Lemma~\ref{lem:lifting} applied to $(**)$, there exists $\tm n \in \satterms$
    such that
    $\sat{\subst{\tm m}}{\subst{(A_1 \OR \dots \OR A_n)}} \proves{\JIKCS}{\sat{\tm n}{A}}$, 
    equivalently (by Definition~\ref{def:substitution})
    $\subst{(\sat{\tm m}{(A_1 \OR \dots \OR A_n)})}\proves{\JIKCS}{\sat{\tm n}{A}}$.
    Hence, $\subst{(\sat{\tm m_1}{A_1} \OR \dots \OR \sat{\tm m_n}{A_n})}\proves{\JIKCS}{\sat{\tm n}{A}}$ where indeed $\sat{\tm n}{A}\in\prel{\F{\DIA[k]H}}$ and $\sub$ is a substitution living on $\DIA[n] H$ and meets the no new variable condition.
\end{proof}

Finally, our central Realisation Theorem falls as a corollary of the results harvested this far:
from pre-realisers which are semantically justified by the Pre-realisation theorem~\ref{thm:quasreal} we extract potential realisers which are also syntactically designed through the Condensing theorem~\ref{thm:condensing} to follow precisely the original structure of the modal formula.
\begin{proof}[Proof of Theorem~\ref{thm:realisation}]
        Suppose $\proves{\IK}{G}$.
        By the Pre-realisation Theorem~\ref{thm:quasreal}, there exist $A_1, \dots, A_n \in \qrel{\F{G}}$ 
        such that 
        $\proves{\JIKCS}{A_1 \OR \dots \OR A_n}$.
        By the Condensing Theorem~\ref{thm:condensing}, there exists $A \in \prel{\F{G}}$ and a substitution $\sub$ such that
        $\proves{\JIKCS}{\subst{(A_1 \OR \dots \OR A_n)} \IMP A}$.
        By the Substitution Lemma~\ref{lem:substitution},
        $\proves{\JIKCS}{\subst{(A_1 \OR \dots \OR A_n)}}$
        and so by modus ponens
        $\proves{\JIKCS}{A}$.
        Let $\real{G} \colonequals A$ and 
        we have $\forget{\real{G}} = G$ by Proposition~\ref{prop:prelforget}.
\end{proof}

\section{Conclusion}\label{sec:concl}

We have presented basic modular and modular models for the justification counterparts to~$\IK$ which let us derive an alternative semantic proof of the realisation theorem from $\IK$ to $\JIK$.
This method could be further adapted on the one hand to 
sublogics of $\IK$ such as the constructive modal logic $\CK$ and on the other to
extensions of $\IK$
such as the other logics in the modal $\ISfive$ cube by following in the classical case~\cite{artemov_logic_1999,rubtsova_evidence_2006,fitting_realization_2011,kuznets_justifications_2012,goetschi_realization_2012,borg_realization_2015},
or the family of Geach/Scott-Lemmon logics~\cite{fitting_modal_2016}.
For example, the modal logic $\ISfour$ extends $\IK$ with the following axioms.
\begin{center}
\begin{tabular}{r @{\hspace{0cm}} c @{\hspace{0.1cm}} l @{\hspace{1cm}} r @{\hspace{0cm}} c @{\hspace{0.1cm}} l @{\hspace{1cm}}  r @{\hspace{0cm}} c @{\hspace{0.1cm}} l @{\hspace{1cm}} r @{\hspace{0cm}} c @{\hspace{0.1cm}} l}
   $\taxb$ & :  &  $\BOX G \IMP G$ & 
   $\taxd$ & :  &  $G \IMP \DIA G$ & 
   $\faxb$ & :  &  $\BOX G \IMP \BOX \BOX G$ & 
   $\faxd$ & :  &  $\DIA \DIA G \IMP \DIA G$
\end{tabular}
\end{center}

%
\noindent Semantically, this means imposing $\ikrel$ is reflexive and transitive in any model $\tuple{\ikworlds, \ikirel, \ikrel,\basicval}$.
In the corresponding justification logic~\cite{marin_justification_2025}, we add the following axioms to $\BaseAx$.
\begin{center}
\begin{tabular}{r @{\hspace{0cm}} c @{\hspace{0.1cm}} l @{\hspace{1cm}} r @{\hspace{0cm}} c @{\hspace{0.1cm}} l @{\hspace{1cm}}  r @{\hspace{0cm}} c @{\hspace{0.1cm}} l @{\hspace{1cm}} r @{\hspace{0cm}} c @{\hspace{0.1cm}} l}
   $\jtaxb$ & :  &  $\just{t}{A} \IMP A$ & 
   $\jtaxd$ & :  &  $A \IMP \sat{m}{A}$ & 
   $\jfaxb$ & :  &  $\just{t}{A} \IMP \just{\jbang{t}}{\just{t}{A}}$ & 
   $\jfaxd$ & :  &  $\sat{m}{\sat{n}{A}} \IMP \sat{n}{A}$
\end{tabular}
\end{center}
%
An adaptation to the basic modular models in Section~\ref{sec:basic} in line with models for the justification counterpart to $\logic{S4}$~\cite{mkrtychev_models_1997,artemov_ontology_2012}, and imposing $\rel$ is reflexive and transitive, should be sufficient.
Similarly, we conjecture that one can find justification counterparts to the family of logics defined by extending $\IK$ with $\DIA^k \BOX^l G \IMP \BOX^m \DIA^n G$ where $k,l,m,n \in \nat$,
by using the frame conditions given in~\cite{plotkin_framework_1986}.
The challenge could be defining additional operations on proof terms and satisfier terms.

Regarding sublogics such as $\CK$, the canonical model construction for modular models would need to be generalised to build over not only prime sets, but pairs of prime sets and \emph{segments} following the methodology introduced by~\cite{wijesekera_constructive_1990} and developed for all the logics between $\CK$ and $\IK$ by~\cite{GroShiClo25}.


Some variants of intuitionistic justification logic can be seen as extending the typed $\lambda$-calculus with reflective capabilities~\cite{artemov_unified_2002}. 
Applications of intuitionistic justification logic could extend to $\lambda$-calculi and type theory~\cite{bellin_extended_2001},
and lead to more expressive modal $\lambda$-calculi, notably for type constructors corresponding to the $\DIA$ modality~\cite{valliappan2025lax}.

%

\bibliographystyle{eptcs}
\bibliography{bibliography}

\newpage
\appendix

\end{document}